\documentclass{article}

\usepackage[accepted]{aistats2022}
\usepackage{hyperref}       
\usepackage{url}            
\usepackage{booktabs}       
\usepackage{amsfonts}       
\usepackage{microtype}      
\usepackage{xcolor}         

\usepackage{subfigure}
\usepackage{booktabs} 

\usepackage{amsmath}
\usepackage{amssymb} 
\usepackage{amsthm}
\usepackage{graphicx,psfrag,epsf, xcolor}
\usepackage{enumerate}
\usepackage[round]{natbib}
\usepackage{url} 
\usepackage{mathtools}
\newtheorem{theorem}{Theorem}[section]
\newtheorem{corollary}{Corollary}[theorem]

\begin{document}

\twocolumn[

\aistatstitle{Policy Learning for Optimal Individualized Dose Intervals}

\aistatsauthor{ Guanhua Chen \And Xiaomao Li \And  Menggang Yu }

\aistatsaddress{ University of Wisconsin, Madison \\ gchen25@wisc.edu \And University of Wisconsin, Madison  \\ gilbert.xm.li@gmail.com \And University of Wisconsin, Madison  \\ meyu@biostat.wisc.edu} 
]


%

\begin{abstract}
	We study the problem of learning individualized dose intervals using observational data. There are very few previous works for policy learning with continuous treatment, and all of them focused on recommending an optimal dose rather than an optimal dose interval. In this paper, we propose a new method to estimate such an optimal dose interval, named probability dose interval (PDI). The potential outcomes for doses in the PDI are guaranteed better than a pre-specified threshold with a given probability (e.g., $50\%$). The associated nonconvex optimization problem can be efficiently solved by the Difference-of-Convex functions (DC) algorithm. We prove that our estimated policy is consistent, and its risk converges to that of the best-in-class policy at a root-n rate. Numerical simulations show the advantage of the proposed method over outcome modeling based benchmarks. We further demonstrate the performance of our method in determining individualized Hemoglobin A1c (HbA1c) control intervals for elderly patients with diabetes.
\end{abstract}

\section{INTRODUCTION}

Increasing evidence has shown that personalized prescriptions based on patients’ health conditions and medical history will
result in more effective treatment \citep{mancinelli2000pharmacogenomics,collins2015new}. Hence, it is desirable to learn personalized treatment policy based on collected data and recommend treatment for future patients based on such a policy. 
%
We focus on off-policy learning, which is about evaluating and optimizing new personalized treatment policies based on observational data.

Despite the abundance of research on policy learning for categorical treatments \citep{Liang2020, Lou2017, Chen2017, kosorok2019precision}, policy learning for continuous treatments has not been studied sufficiently. Meanwhile, many applications in the medical domain involve determining optimal continuous treatment: the optimal dosage of a medical drug, e.g., Warfarin \citep{international2009estimation, rich2014simulating}, and the optimal length of stay in the hospital for hospitalized patients \citep{borghans2012length}. Existing approaches for continuous treatments focused on recommending a fixed optimal dose for each patient. However, for many medical applications, recommending optimal dose interval can be more practical. For example, for chronic disease management metrics (e.g. HbA1c level, blood pressure, and cholesterol), targeted intervals, instead of a particular level, are frequently recommended from major medical organizations such as the American Diabetes Association \citep{american2019management} and the American Heart Association \citep{flack2020blood}. Another example is cancer radiation therapy, where radiation received by targeted tumors, not delivered by physicians, can be controlled only at an interval due to the irregular shape of solid tumors and the need to radiate surrounding regions \citep{scott2017genome}. Therefore we focus on policy learning for optimal recommendation of an interval instead of a single dose. In addition, this shifted focus will help relax the local approximation needed for optimal dose finding, therefore improving both computational stability and theoretical properties, including a faster convergence rate for the target function over an existing method \citep{chen2016personalized}.

\section{PROBLEM SETUP}
\subsection{Settings and Assumptions}
Denote the dataset with $n$ observations by $(X_i, A_i, Y_i)$, $i = 1, \ldots, n$. Here, $X_i \in \mathcal{X}$ is a $d$-dimensional vector representing patient-level covariates. $A_i \in \mathcal{A}\equiv [a_L, a_U]$ is the observed treatment dose whose range has the lower and upper bounds ($a_L$ and $a_U$). $Y_i$ is the observed outcome from receiving $A_i$. We adopt the well-known potential outcome notation in causal inference \citep{imbens2015causal} to articulate our research problem. For any dose $a\in \mathcal{A}$, the potential outcome $Y_i(a)$ is the outcome that would have been observed if patient $i$ received treatment $a$. Note that when $A_i = a$, naturally $Y_i = Y_i(a)$ by the causal consistency assumption \citep{imbens2015causal}. An individualized policy or dose assignment rule is a functional $f(X): \mathcal{X} \rightarrow \mathcal{A}$. For simplicity of presentation, we assume bigger $Y(a)$ is better. The optimal policy is the one that maximizes the policy value $V(f)\equiv E[Y(f(X))]$. To estimate this causal quantity using the observed data, we adopt the following three commonly assumed conditions \citep{imbens2015causal}: 1) the stable unit treatment value assumption (SUTVA); 2) the no unmeasured confounding assumption: $Y(a)\perp A| X$; and 3) the positivity assumption which assumes $P(A=a|X) >0$ for all $a\in \mathcal{A}$. $P(A=a|X)$ is known as the (generalized) propensity score function that models the treatment assignment mechanism  \citep{imbens2015causal}. 

An Individualized Dose Interval (IDI) is desirable for a patient if the doses inside the interval are more likely associated with better outcomes than doses outside the interval. Each IDI is therefore connected with a threshold $S$ that separates good outcomes from bad ones for the patient. In addition, a confidence probability $\alpha$ reflects how strongly the desirable interval is associated with the good outcomes. Specifically, let $f_L(x)$ and $f_U(x)$ be the lower and upper bounds of the IDI for a patient with covariates $x$ such that $a_L\le f_L(x)\le f_U(x)\le a_U,\;\forall x\in \mathcal{X}$.
Then, we define a level-$\alpha$ \textit{Probability Dose Interval} (PDI) as follows:
\begin{equation} \label{eq:pdi-define}
	\begin{aligned}
		&	\text{PDI}_{\alpha}(x) \equiv  \big[f_L(x),f_U(x)\big]: \ \forall  a\in[f_L(x),f_U(x)], \\
		&	P\{Y(a)>S|x\} \geq \alpha. 
	\end{aligned}
\end{equation}
If a clinician follows the PDI recommendation when treating a patient with covariates $x$, the patient is guaranteed to have the outcome above $S$ with a probability of at least $\alpha$.

\subsection{Related work}
\textbf{Indirect methods}: Indirect methods for estimating treatment rules are based on outcome modeling \citep{zhao2009reinforcement, moodie2012q, chakraborty2013statistical}. This framework has two steps. The first step models the conditional expectation of the outcome given the treatment and covariates, $E(Y|A, X)$. The models can be parametric such as regression with Lasso-type regularization \citep{tibshirani1996regression},  or nonparametric such as regression tree \citep{breiman1984classification,breiman2001random} and Support Vector Machine \citep{steinwart2008support,vapnik2013nature}. In the second step, the treatment option which maximizes the conditional expection in the first step is set as the optimal treatment. For most cases, the second step involves numerical optimization such as grid search due to lack of close-form solutions. Since there are no existing approaches for PDI learning, we propose the following indirect PDI learning approach to serve as a benchmark. In particular, the first step for indirect PDI learning would construct an outcome model to estimate the conditional probability $P\big(Y>S\,|\,A=a,X=x\big)$, for example, based on a classifier for predicting the binary label $I\{Y>S\}$ using $A$ and $X$. Subsequently, a grid search for doses is done  to identify the dose interval where the predicted conditional probability is greater than the confidence level  $\alpha$. It has been noted in binary treatment settings that solutions from the indirect methods can be unstable due to the need to model  $P\big(Y>S\,|\,A=a,X=x\big)$ for all $a$ and $x$ \citep{kosorok2019precision}.  Furthermore, the grid search step could be computationally challenging.

\textbf{Direct methods}: On the other hand, for optimal policy learning with categorical treatments, direct methods have been developed to learn the optimal policy directly, without using the outcome modeling. These approaches \citep{moodie2009estimating, Xu2015, Lou2017,Chen2017,Liang2020} are believed to have better performance when signal-to-noise ratio is low. For learning optimal policy for continuous treatments, \cite{chen2016personalized} proposed an outcome weighted learning (O-learning) method which is based on a weighted regression. In particular, O-learning maximizes the value function associated with a policy $f$: i.e. $f_{opt} = \text{argmax}_{f} V(f)$ with $V(f)=E[Y(f(X))]$. Furthermore, \cite{chen2016personalized} and other direct methods \citep{kallus2018policy, bica2020estimating, sondhi2020balanced, zhu2020kernel, schulz2020doubly,zhou2021parsimonious} use inverse probability weighted (IPW) or augmented IPW estimators of the value function to address confounding issues by reweighting the observational data to mimic randomized trials. However, none of these approaches are directly applicable for the PDI problem.

\section{DIRECT POLICY LEARNING FOR PDI}\label{Estimation}
Because the recommended dose interval needs to satisfy probability constraint $P\big(Y (a)>S\,|\,X=x\big) \geq \alpha$ for $a \in \text{PDI}_{\alpha}(x)$, maximizing the associated value function for a dose interval becomes less intuitive. Instead, we propose a loss function such that the estimated optimal policy that minimizes the loss function would (asymptotically) satisfy the constraint.

There are two types of errors that an PDI can incur. When a patient receives a dose inside the PDI, $a \in\text{PDI}_{\alpha}(x)$, but has an outcome below the threshold, $Y(a)<S$. It presents a false-positive error. On the other hand, when the patient receives a dose outside the interval, $a \not \in\text{PDI}_{\alpha}(x)$, but has an outcome above the threshold, $Y(a)>S$. It presents a false-negative error. The goal is to minimize these two errors. Instead of working with a multi-objective optimization problem, we use $\alpha$ to balance the magnitudes of the two errors and work with the weighted sum, that is, 
$\alpha I\{Y(a) \leq S, a\in\text{PDI}_{\alpha}(x)\} + (1-\alpha)I\{Y(a) > S, a\not \in\text{PDI}_{\alpha}(x)\}$. 
This weighted sum leads to nice theoretical guarantee of the resulting interval for our proposed method as we demonstrate in Section \ref{Theory} below.

The indirect methods rely on modeling the outcome $Y$ to predict $I(Y(a)>S)$. Instead, we circumvent modeling of the outcome and directly minimize the above error based on observed data. Specifically, we consider the following risk function for minimization to obtain an estimate for $\text{PDI}_{\alpha}(x)$ with any given $\alpha$, and $S$. 
\begin{equation}
	\begin{aligned}  \label{risk-PDI}
		& E\bigg[\Big\{\alpha I(Y\leq S) I(A\in \text{PDI}_{\alpha}(X)) \\ 
		& +(1-\alpha) I(Y>S) I(A\notin \text{PDI}_{\alpha}(X))\Big\} \frac{1}{P(A\,|\,X)}\bigg]
	\end{aligned}
\end{equation}
Here $P(A\,|\,X)$ is the generalized propensity score account for possible nonrandom treatment assignment  in observational studies. 

Under mild conditions, the minimizer of the risk function \eqref{risk-PDI} becomes an interval $[f_L(x),f_U(x)]$ where $f_L(x)$ and $f_U(x)$ are boundary functions that minimize the following risk
\begin{equation}
	\begin{aligned}
		& R(f_L,f_U) = E\bigg[\Big\{\alpha I(Y\leq S) I(A\in [ f_L(X),f_U(X) ])  \\
		&  +(1-\alpha)I(Y>S) I(A\notin [f_L(X),f_U(X)])\Big\} \frac{1}{P(A\,|\,X)}\bigg]\label{two-sided PDI}
	\end{aligned}
\end{equation}

In cases when Conditions C1 and C2 defined in Section~\ref{Theory} are satisfied, then only the lower bound is needed to separate the desirable doses from the undesirable doses, i.e., $\text{PDI}_{\alpha}(x)=[f_L(X),a_U]$. Accordingly, we just need to minimize the following to estimate $f_L(x)$.
\begin{equation}
	\begin{aligned}
		&	R(f_L) =	E\bigg[\Big\{\alpha I(Y \leq S) I\big(A\in[f_L(X),a_U]\big) \\
		&	 +(1-\alpha)I(Y>S) I\big(A\in [a_L,f_L(X)]\big)\Big\}\frac{1}{P(A\,|\,X)}\bigg]\label{one-sided PDI}
	\end{aligned}
\end{equation}

The functional forms for the boundary functions $f_L(x)$ and $f_U(x)$ can be chosen quite flexibly in principle, as in many machine learning approaches. In practice, one may put some 
some smoothness conditions on them for given $x$ to ensure that a small perturbation of $x$ will not lead to a dramatic change in the resulting PDI.

\subsection{Non-Convex Surrogate Loss Relaxation}\label{Non-Convex Relaxation}

As in the classical classification tasks, it is difficult to directly optimize the empirical risk due to the discontinuity of indicator functions in \eqref{two-sided PDI} and \eqref{one-sided PDI}. Surrogate loss functions have been proposed to replace the indicator functions as long as the surrogate loss functions satisfy certain conditions such as Fisher consistency. Note that unbounded loss functions such as the hinge loss are not good options here as demonstrated in the optimal dose problem \citep{chen2016personalized}. The reason is due to the sensitivity of the solution of the unbounded surrogate loss to outliers, that is, to observations whose received doses are far from the optimal bounds. Therefore we propose a truncated hinge loss as a surrogate similar to \cite{chen2016personalized}. Such non-convex surrogate losses have also been used in robust SVM \citep{wu2007robust}.  

The risk function \eqref{one-sided PDI} of the one-sided $\text{PDI}_{\alpha}(X)$ after relaxation is as follows. See Figure~\ref{Psi} for a graphical representation.
\begin{equation}\label{relax_1sided}
	\begin{aligned}
		& R_{\epsilon}(f_L)=E \bigg[\Big\{ \alpha I\big(Y \leq S\big) \Psi_\epsilon(A,{f}_L) \\
		& +\big(1-\alpha\big)I\big(Y > S \big)  \Psi_\epsilon({f}_L ,A) \Big\} \frac{1}{P(A\,|\,X)}\bigg]
	\end{aligned}
\end{equation} 
where $\Psi_\epsilon(a,b)=min\{{\epsilon}^{-1} {(a-b)_+},1\}.$  Note that  $\Psi_\epsilon(a,b)= \Psi_{\epsilon,1}(a,b) - \Psi_{\epsilon,2}(a,b)$, where  $\Psi_{\epsilon,1}(a,b)=\{{\epsilon}^{-1}{(a-b)}\}_+$ and $\Psi_{\epsilon,2}(a,b)=\{{\epsilon}^{-1}{(a-b)}-1\}_+$ are convex.

The risk function \eqref{two-sided PDI} of the two-sided $\text{PDI}_{\alpha}(X)$ after relaxation is as follows.
\begin{equation}\label{relax_2sided}
	\begin{aligned}
		& R_{\epsilon}(f_L,f_U) =E \bigg[\Big\{  \alpha I \big(Y \leq S \big) \Psi_\epsilon^{in}\big({f}_L ,A,f_U \big) \\ 
		& +  \big(1-\alpha\big) I\big(Y > S\big) \Psi_\epsilon^{out}\big({f}_L,A,f_U \big) \Big\}  \frac{1}{P(A\,|\,X)}\bigg]
	\end{aligned}
\end{equation}
where 
$$\Psi_\epsilon^{in}\big(a,b,c\big)=
\begin{cases}
	0, & b<a<c\\
	(b-a)/\epsilon, & a<b<a+\epsilon<c\\
	1, &a+\epsilon<b<c-\epsilon\\
	(c-b)/\epsilon, &a<c-\epsilon<b<c\\
	0, &a<c<b \\
\end{cases}
\\
\text{, and} 
$$ 
$\Psi_\epsilon^{out}\big(a,b,c\big)=1-\Psi_\epsilon^{in}\big(a,b,c\big).$

\begin{figure*}[!ht]
	\centering
	\includegraphics[width=10cm]{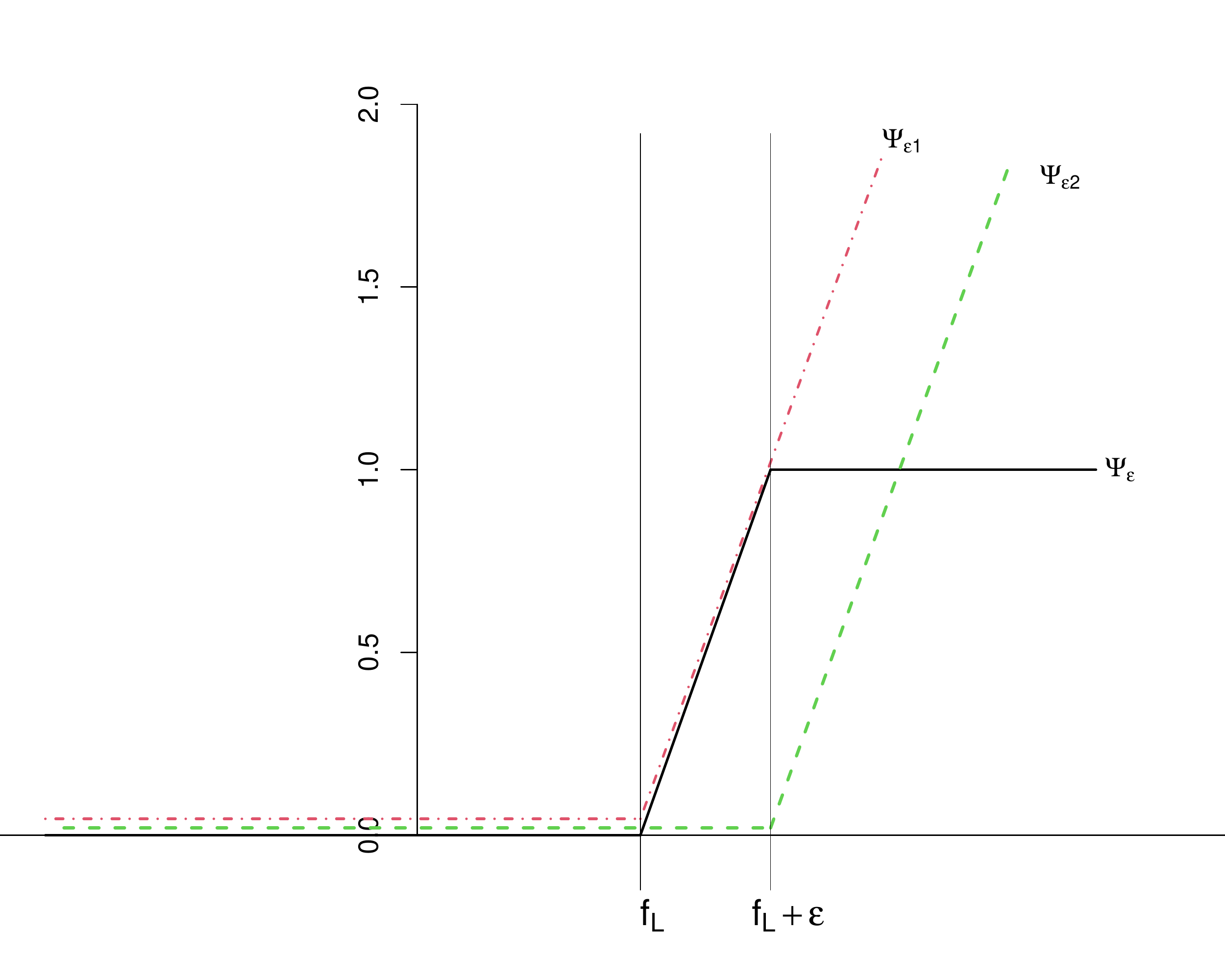}
	\caption[Decomposition of Non-convex Loss Function $\Psi_{\epsilon}$]{Loss Functions $\Psi_\epsilon$,$\Psi_{\epsilon, 1}$ and $\Psi_{\epsilon,2}$} \label{Psi}
\end{figure*}

\subsection{DC Algorithm for One-sided PDI}\label{DC}
Assume that $f_L$ belongs to a reproducing kernel Hilbert space (RKHS) $\mathcal{H}$ \citep{shawe2004kernel}. We propose minimizing the following regularized empirical risk of $R_{\epsilon}(f_L)$ corresponding to  \eqref{relax_1sided}: 
\begin{align} %
	\hat{f}_{L,n} &= \arg\min_{f_L \in \mathcal{H}}\sum_{i=1}^n\bigg[\Big\{ \alpha I(Y_i \leq S)\Psi_\epsilon(A_i, f_L(X_i)) \nonumber \\ 
	& + \big(1-\alpha\big)I(Y_i > S)\Psi_\epsilon( f_L(X_i),A_i) \Big\} \frac{1}{P(A_i|X_i)}\bigg]  \nonumber \\ &+ \frac{\lambda_n}{2}\|f_L\|_{\mathcal{H}},  \label{regularized-empirical-risk}
\end{align}
where the regularization term  $\frac{\lambda_n}{2}\|f_L\|_{\mathcal{H}}^2$  penalizes the complexity of $f_L$ and $\| \cdot \|_{\mathcal{H}}^2$ is a norm defined in $\mathcal{H}$.  

Optimizing such an objective function is challenging due to the nonconvexity of the loss function. We adopt the Difference-of-Convex-functions (DC) algorithm \citep{le1997solving} to iteratively solve a series of convex minimization problems. For simplicity, assume $f_L(X_i; \mathbf{v}) = \sum_{j=1}^{n}v_j k(X_i, X_j) + v_0$ where $k$ is the reproducing kernel  of $\mathcal{H}$ and denote $\mathbf{v} = (v_1,.....,v_n, v_0)$ for the unknown parameters. We can write the regularized empirical risk in \eqref{regularized-empirical-risk} as
\begin{align}
	\mathcal{L}(\mathbf{v}) & \equiv \sum_{i=1}^n\bigg[\Big\{ \alpha I(Y_i \leq S)\Psi_\epsilon\big(A_i, f_L(X_i; \mathbf{v}) \big) \nonumber \\
	& +\big(1-\alpha\big) I(Y_i > S) \Psi_\epsilon\big( f_L(X_i; \mathbf{v}) ,A_i\big) \Big\} \frac{1}{P(A_i|X_i)}\bigg] \nonumber\\
	&+ \frac{\lambda_n}{2} \mathbf{v} ^\top \mathbf{K}  \mathbf{v} \label{regularized-emp-risk2}
\end{align}
where $\mathbf{K}$ is a $n \times n$ matrix with $k(X_i,X_j) $ ($ i = 1 \ldots n; j = 1, \ldots, n$) as its elements. Let $\mathcal{L}_1(\mathbf{v})$ and $\mathcal{L}_2(\mathbf{v})$ be similarly defined as $\mathcal{L}(\mathbf{v})$ except with $\Psi_\epsilon$ replaced with $\Psi_{\epsilon,1}$ and $\Psi_{\epsilon,2}$ (without the regularization term) respectively in \eqref{regularized-emp-risk2}. Then $\mathcal{L}(\mathbf{v})=\mathcal{L}_1(\mathbf{v})-\mathcal{L}_2(\mathbf{v})$ and the non-convex objective function $\mathcal{L}(\mathbf{v})$ is the difference of two convex functions $\mathcal{L}_1(\mathbf{v})$ and $\mathcal{L}_2(\mathbf{v})$. By replacing the second function $\mathcal{L}_2(\mathbf{v})$ by its first order Taylor expansion expanded around the solution from previous iteration and the fact that a linear function is convex and concave, the DC algorithm solves a convex optimization problem at each iteration.

In summary, the DC algorithm estimates $\mathbf{v}$ by repeatedly solving the following optimization problem
\begin{eqnarray*}
	\mathbf{v}^{t+1} & =\arg\min_\mathbf{v}\hskip0.01\textwidth \mathcal{L}_1(\mathbf{v})-\{\nabla \mathcal{L}_2(\mathbf{v}^t)\}^\top(\mathbf{v}- \mathbf{v}^t) \\
	&	= \arg\min_\mathbf{v} \hskip0.01\textwidth \mathcal{L}_1(\mathbf{v})-\{\nabla \mathcal{L}_2(\mathbf{v}^t)\}^\top \mathbf{v}, \nonumber 
\end{eqnarray*}
where $\mathbf{v}^{t+1}$ and $\mathbf{v}^{t}$ are solutions of the DC-algorithm at $t+1$-th and $t$-th iterations, respectively.
The optimization problem above is a quadratic programming problem which can be solved by standard packages. The details of the derivations are documented in the supplementary. The convergence property of the DC algorithm has been shown in \cite{lanckriet2009convergence} such that the solution is guaranteed to converge to a local minimum. 

So far, we only discussed the case of one-sided PDIs with a lower bound function. In the case where an upper bound is desired, one could also utilize the algorithm described above with a minor modification. By taking the negative value of the dose, the upper bound of the original dose will become the lower bound for the negative dose.
To find two-sided PDIs, one can use the following strategy to convert it into two separate one-sided PDI problems. Recall that the optimal dose function $f_{opt}=\arg\max_{f}V(f)$  defined in Section $2.2$ and assume that $f_{L,opt}(x) < f_{opt}(x)< f_{U,opt}(x)$ for all $x\in \mathcal{X}$. Then, we can use samples with observed $A < f_{opt}(x)$ to learn the lower bound and samples with observed $A > f_{opt}(x)$ to learn the upper bound. To estimate the unknown function $f_{opt}(x)$, we can use the existing methods for optimal dose finding \citep{chen2016personalized, kallus2018policy, zhu2020kernel, zhou2021parsimonious}. The choice of which PDI (one-sided or two-sided) to use usually depends on practical considerations. When recommending HbA1c level for diabetes patients, because over-control HbA1c is the main concern, learning a one-sided PDI is sufficient (see Section~\ref{RealData} for details). In other settings, e.g. optimal personalized dose intervals for the Warfarin drug (an anticoagulant drug), a two-sided PDI may be preferred as overdosing on Warfarin will increase risk of blooding, and underdosing will increase risk of stroke. 

Note that PDIs depend on the threshold $S$ and the probability $\alpha$. The combinations of these quantities with different values could lead to different interpretations of the resulting PDIs. The pre-determined threshold $S$ that defines the favorable outcomes can also be chosen as dependent on some covariates $z\subset x$, e.g., age or age groups. 

\subsection{Generalized Propensity Score for Continuous Treatment}
\label{section:GPS}

To conduct off-policy learning for a continuous treatment, we need to estimate $P(a\,|\,x)$ which is the conditional density of dose $a$ given the covariates $x$ \citep{austin2018assessing} and weigh each observation using the inverse of the estimated $P(a\,|\,x)$. Multiple approaches are available for estimating $P(a\,|\,x)$. One natural approach is regression. In particular, one could assume the density is from parametric or nonparametric distributions \citep{austin2018assessing}. This approach usually performs satisfactorily when the distributional assumptions are reasonable. For our simulation settings as well as the data application, the approach led to quite unstable inverse weights. To remedy the problem, we adopt recent methods that directly estimate the weights (instead of the inverse of estimated $P(a\,|\,x)$) that could balance the covariates across the range of a continuous treatment \citep{fong2018covariate, kallus2019kernel, tubbicke2020entropy, huling2021independence}. In particular, we use a modified version of Kernel Optimal Orthogonality Weighting \citep{kallus2019kernel} and the independence-inducing weights by \cite{huling2021independence} for our numerical analysis due to their stable finite sample performance.

\section{THEORETICAL RESULTS}\label{Theory}

We include results for one-sided PDIs here and leave those for two-sided PDIs to the supplement. If $[f_L(x), a_U]$ is a one-sided level-$\alpha$ PDI, then it is reasonable to assume $[f_L(x) + \eta(x), a_U]$ is also a one-sided level-$\alpha$ PDI for any $\eta(x) \ge 0$ as it is a narrower interval. We ensure this by making the following assumptions throughout the section:

(C1) For any $x$, the boundaries $a_L$ and $a_U$ for the dose range satisfy $P\{Y(a_L)>S \big |x\}\leq \alpha$ and $P\big\{Y(a_U)>S |x \big\}\geq  \alpha$.\\ [8pt]
(C2) For any $x$, there is a unique solution $a_x \in [a_L,a_U]$ such that  $P\{Y(a_x)>S |x\} = \alpha$.

The first condition (C1) ensures that our choices of $S$ and $\alpha$ are reasonable in the sense that the level-$\alpha$ is attainable. The second condition (C2) ensures that the limit of our learning algorithm is well defined. Details of the proof can be found in the supplementary file. Given (C1) and (C2), our solution $f_L(x)$ can attain the lowest possible value (i.e., the PDI is as wide as possible). 

Note that if $P(Y > S| A = a, X=x)$ is a monotonic function in $a$ for all $x$, i.e.
$P(Y > S | A =a_1, X=X ) \leq P(Y > S | A= a_2, X=X)$, $\forall a_1,a_2  \in  \text{PDI}_{\alpha}(X)$ and $a_1 < a_2$, then Conditions (C1) and (C2) will hold with properly chosen $\alpha$ and $S(X)$. However, (C1) and (C2) do not require  $P(Y(a) > S|X=x)$ to be monotonic. In other words, the monotonicity is  sufficient but not necessary for (C1) and (C2) to hold. We make assumptions that (C1) and (C2) hold to guarantee the existence of $\text{PDI}_{\alpha}(X)$. Hence, our method can work with multi-modal outcomes as long as (C1) and (C2) are satisfied. On the other hand, for many medical applications, one can safely expect $P(Y>S|a,X)$ to be  monotonic or unimodal in $a$.  A unimodal $P(Y>S|a,X)$ will require two-sided PDIs. 


\subsection{Fisher Consistency}\label{consistency} Theorem \ref{THMconsist-p-1} below shows that minimizing $R(\cdot)$ defined in \eqref{one-sided PDI} will result in a one-sided PDI lower bound, above which the doses will produce outcomes larger than $S$ with a probability of at least $\alpha$. 


\begin{theorem}\label{THMconsist-p-1}
	Let $f_{L,opt}= \arg\min_f R(f)$. Then, under the assumptions (C1) and (C2), $f_{L,opt}(x)$ has the following two properties:\\ [8pt]
	(1) $P\big\{Y\big(f_{L,opt}(x)\big)> S \, |\, x \big\}= \alpha, \forall x$. \\ [8pt]
	(2) For any measurable function $a(\cdot)$ s.t. $f_{L,opt}(x)\leq a(x)\leq  a_U$,  $\forall x\in \mathcal{X}$, the potential outcome $Y(a(x))$ satisfies $P(Y(a(x)) > S|x)\geq \alpha$.
\end{theorem}

Theorem \ref{THMconsist-p-1} shows that for any given value of $X=x$, the length of $[f_{L,opt}(x),a_U]$ is the widest among all valid $\text{PDI}_{\alpha}(x)$. 

\subsection{Convergence Rate}\label{Convergence Rate}
Theorem \ref{THMbound-p-1} below shows that, for one-sided PDI lower bounds, the difference between the risk function $R(\cdot)$ and its non-convex relaxation $R_{\epsilon_n}(\cdot)$ converges to zero with the same rate as $\epsilon_n\rightarrow 0$. That is, the approximation bias goes to $0$. Then in Theorem \ref{THMoptimalrate-1}  we derive the convergence rates of the risk function based on the estimated $\hat{f}_{L,n}$ and the general results of empirical risk minimization and approximation. This theorem tells us how the empirically optimal policy performs out-of-sample. Here, we assume the Gaussian kernel with bandwidth $\gamma_n$ (i.e. $\mathbf{K}(x,x') = \exp(-\gamma_n^2||x - x'||^2))$ was used in estimating $\hat{f}_{L,n}$. Since ${f}_{L,opt}$ may not belong to the RKHS we optimize in, we choose the Gaussian kernel (which is dense in $L_2$) to make sure $\hat{f}_{L,n}$ can well approximate ${f}_{L,opt}$.

\begin{theorem}\label{THMbound-p-1}
	For any measurable function $f:\mathcal{X}\mapsto \mathbb{R}$, we have $|R(f)-R_{\epsilon_n}(f)|\leq C\epsilon_n$, where $C$ is a constant. 
\end{theorem}

\begin{theorem}\label{THMoptimalrate-1} 
	Let $B^{\delta}_{1,\infty}(\mathbb{R}^d)$ be a Besov space defined as $B^{\delta}_{1,\infty}(\mathbb{R}^d)=\{f\in L_{\infty}(\mathbb{R}^d):\sup_{t>0} (t^{-\delta}w_{(r,L_1)}(f,t))<\infty\}$, where $w$ is the modulus of continuity.  Assume that $f_{L,opt}\in B^{\delta}_{1,\infty}(\mathbb{R}^d)$. Then for any $\zeta>0$, $d/(d+\tau)<p<1$, and $\tau>0$, 
	\begin{eqnarray*}
		&	 R(\hat f_{L,n})-R(f_{L,opt})\leq c_1\frac{\lambda_n}{\gamma_n^d}+c_2\gamma_n^\delta  +  \\ & c_3{\gamma_n^{-\frac{(1-p)(1+\zeta)d}{2-p}}\lambda_n^{-\frac{p}{2-p}}n^{-\frac{1}{2-p}} }
		+c_4\frac{\tau^{\frac{1}{2}}}{n^{\frac{1}{2}}}+c_5\frac{\tau}{n}+c_6\epsilon_n
	\end{eqnarray*}
	with probability $1-3e^{-\tau}$. Furthermore, if we choose these constants to be
	\begin{align*}
		\gamma_n =\mathcal{O}\bigg(\frac{1}{n}\bigg)^{\frac{1}{2\delta+d}}\hspace{0.01\textheight}\lambda_n =\mathcal{O} \bigg(\frac{1}{n}\bigg)^{\frac{\delta+d}{2\delta+d}}\hspace{0.01\textheight}\epsilon_n=\mathcal{O}\bigg(n^{-\frac{\delta}{2\delta+d}}\bigg).
	\end{align*}
	We have the following optimal convergence rate:
	$
	R(\hat f_{L,n})-R( f_{L,opt})=\mathcal{O}\bigg(n^{-\frac{\delta}{2\delta+d}}\bigg).
	$
\end{theorem}

Theorem \ref{THMoptimalrate-1} shows that the O-learning for PDI can potentially achieve a rate of nearly $\mathcal{O}(n^{-\frac{1}{2}})$, when the smoothness parameter $\delta$ of the optimal bound function $f_{L,opt}$ is relatively large compared to the covariate dimension $d$. This rate is faster than the convergence rate of about $\mathcal{O}(n^{-\frac{1}{4}})$ in \cite{chen2016personalized}, which focused on finding the optimal dose. This improvement is due to the fact that more sample points are contained in the neighborhood of the entire optimal interval than in the neighborhood of the optimal dose. Note that, we have assumed the generalized propensity score $P(a|x)$ can be  estimated at the rate of $\mathcal{O}(n^{-1/2})$ or faster. Consequently, the estimation error of the generalized propensity score is negligible. However, if the generalized propensity score can not be estimated at this rate, then convergence rate in Theorem \ref{THMoptimalrate-1} will be dominated by the possibly slower convergence rate of the generalized propensity score estimator.

\section{EXPERIMENTS}\label{exp}
\subsection{Simulation}\label{Simulation}
We compare our methods with three indirect methods for learning one-sided $\text{PDI}_{\alpha}(X)$. Note that existing optimal dose learning methods can not yield a dose interval, hence they are not included in the comparison. For our methods, we implemented linear and Gaussian kernels and denote them as L-O-Learning and G-O-Learning. As the DC-algorithm is an iterative algorithm which relies on initial values, we used covariates-independent lower bounds as such initial values. All indirect methods involve classification of outcomes to find the PDIs. In particular, the indicator $I{\{Y>S\}}$ was used as the binary outcome for classification so that we could build a model for $\hat P(Y>S|a,X) \equiv h(A,X)$. Then we found the PDIs as a range of doses $\mathcal{A}_1(X)$ such that for $\forall a\in \mathcal{A}_1(X)$,  $h(a,X)>\alpha$. To find $\mathcal{A}_1(X)$, we  evaluated different $a$ value in estimated $h(a,X)$. All $a$ that satisfied $h(a,X)>\alpha$ were included in $\mathcal{A}_1(X)$. In our simulations, we used $200$ evenly spaced grids to construct $\mathcal{A}_1(X)$. 

We used logistic regression with $L_1$ penalty, Support Vector Machine (SVM), and Random Forests (RF) as indirect methods. For the penalized logistic regression, we included linear and quadratic terms of the dose, covariates, and pairwise interactions between dose and covariates as predictors. For SVM, we used the Gaussian kernel with the kernel bandwidth fixed using the ``median heuristic". For RF, we set the number of trees to be $1000$. All methods required parameter tuning and we tuned the $L_2$ regularization parameters for O-learnings, the $L_1$ regularization parameter for the logistic regression, the margin trade-off parameter for the SVM, and the percentage of variables that can be used for splitting at each node for the RF. Tuning parameters for all methods were selected using 5-fold cross-validation. 

We considered multiple simulation settings with various data generating processes, sample sizes, numbers of covariates as well as signal to noise levels. We set the confidence probability $\alpha$ to be $0.5$. Scenario $1$ is for linear lower bounds and Scenario $2$ for nonlinear lower bounds. For both scenarios, there is a unique value of $a$ such that $P(Y(a) > S(x)|x) = \alpha$ for a given $x$. 

We assume that the training data contain $n$ i.i.d. observations $(A_i,Y_i,X_i), i = 1 \ldots n$. The dose $A_i$ is generated from a truncated normal distribution $\mathcal{N}(\mu_A,0.5)$ with the truncation limits $[-2,2]$. $X_i$ is generated from a uniform distribution $\text{Uniform}(-1,1)^{d}$. The outcome $Y$ is generated from a normal distribution $\mathcal{N}(\mu_Y,\sigma^2)$. We assume there exist confounders (i.e., $X$ can impact both treatment assignment and the outcome) to mimic the observational study settings. In the supplement, we also include simulations for randomized trial settings. 

\noindent{\textbf{Scenario 1: Linear bounds} } \\
$\mu_A=0.3X_{1}+0.3X_{2}+0.3X_{3};$  \\
$\mu_Y= {5e^{10A}}/(e^{10\mu_A}+e^{10A})+D(X);$ \\
$D(X)=0.6X_{2}+0.6X_{3}+0.6X_{4}.$

\noindent{\textbf{Scenario 2: Non-linear bounds}} \\
$\mu_A=0.75\log(|X_{1}|+1)-0.2\cos(\pi X_{2}) + 0.2I(X_{3}>0) - 0.4;$  \\
$\mu_Y={5e^{10A}}/(e^{10C}+e^{10A})+D(X);$ \\
$D(X)=0.4 \sin(\pi X_{2}) + 0.4 I(X_{3}>0)+0.4|X_{4}|.$ 

We repeat each setting for $100$ times and report the results of all methods by evaluating the empirical risk or the empirical value of $R(\hat{f}_L)$ defined in Equation \eqref{one-sided PDI} on corresponding independent testing data sets with sample sizes of $10000$. For O-learning methods, we estimate weights that could approximate $1/P(a|x)$ using the approaches mentioned in the Section \ref{section:GPS}. Meanwhile when evaluating the performances of all methods, we plug in the true values of $1/P(a|x)$ for calculating the empirical risks of estimated policies on the testing data sets.

\begin{table*}[ht]
	\centering
	\caption{Empirical risk ($R(\hat{f}_L)$) with SD in (.) from comparison methods under different settings} \label{simresults}
	\begin{tabular}{clllllll}
		\hline
		& n  & d & L-O-Learning & G-O-Learning& Logistic & SVM & RF \\ 
		\hline
		Scenario 1& 200 & 10  & \textbf{0.048 (0.004)}  & 0.049 (0.004) & 0.054 (0.023) & 0.053 (0.005) & 0.053 (0.005) \\ 
		$\sigma^2 = 2.25$		   & 200 & 50  &  0.055 (0.005) & \textbf{0.052(0.002) }& 0.176(0.074) & 0.077 (0.041) & 0.058 (0.004) \\ 
		& 400 & 10  & 0.045 (0.002)  & 0.047 (0.002) & \textbf{0.044(0.002) }& 0.048 (0.002) & 0.049 (0.002) \\ 
		& 400 & 50  &  0.056 (0.004) & \textbf{0.051 (0.002) }& 0.071(0.025) & 0.055 (0.003) & 0.053 (0.002) \\ 
		\hline
		Scenario 1 & 200 & 10  & 0.122 (0.004)  & \textbf{0.121 (0.003) } & 0.182(0.049) & 0.139 (0.026) & 0.134 (0.010) \\ 
		$\sigma^2 = 9$				& 200 & 50  &  0.132 (0.005) & \textbf{0.122(0.002) }& 0.215(0.049) & 0.238 (0.032) & 0.142 (0.014) \\ 
		& 400 & 10  & \textbf{0.120 (0.003)}  & 0.121 (0.003) & 0.142(0.037) & 0.124 (0.007) & 0.128 (0.005) \\ 
		& 400 & 50  &  0.129 (0.003) & \textbf{0.122 (0.002) }& 0.205 (0.049) & 0.139 (0.023) & 0.127 (0.004) \\ 
		\hline
		Scenario 2 & 200 & 10 & 0.054 (0.003) & \textbf{0.050 (0.002)} & 0.069 (0.044) & 0.055 (0.005) & 0.052 (0.005) \\ 
		$\sigma^2 = 2.25$				& 200 & 50& 0.061 (0.005) & \textbf{0.051 (0.002}) & 0.187 (0.078) & 0.177 (0.081) & 0.057 (0.006) \\ 
		& 400 & 10  & 0.051 (0.002)  & \textbf{0.049 (0.002) }& 0.051 (0.010) & 0.051 (0.004) & 0.049 (0.005) \\ 
		& 400 & 50  &  0.061 (0.003) & \textbf{0.051 (0.002)} & 0.095(0.050) & 0.060 (0.003) & 0.051 (0.003) \\ 
		\hline
		Scenario 2  & 200 & 10 & 0.123 (0.003) & \textbf{0.120 (0.003}) & 0.177 (0.052) & 0.144 (0.035) & 0.131 (0.011) \\ 
		$\sigma^2 = 9$   	& 200 & 50 & 0.133 (0.005) & \textbf{0.121 (0.002}) & 0.199 (0.053) & 0.249 (0.010) & 0.145 (0.016) \\ 
		& 400 & 10  & 0.121 (0.003)  & \textbf{0.120 (0.002) }& 0.143 (0.038) & 0.123 (0.005) & 0.127 (0.009) \\ 
		& 400 & 50  &  0.131 (0.004) & \textbf{0.121 (0.002)} & 0.198 (0.053) & 0.153 (0.040) & 0.125 (0.005) \\ 
		\hline
	\end{tabular}
\end{table*}

For all methods we compared, we set $S$ to be the fitted value of a polynomial regression of Y with X as predictors (i.e., set $S(x)=\hat E\left[Y|x\right]$), as there are strong prognostic effects. From Table~\ref{simresults}, the L-O-Learning has advantages in the linear bound settings and the G-O-Learning has advantages in the nonlinear settings. In particular, compared to other methods, our methods are more robust to high covariate dimensionality (large $d$), low sample size (small $n$), and low signal-to-noise ratio (large $\sigma^2$) settings. In terms of computational time, the RF and L-O-Learning/G-O-Learning are comparable, and they are much faster than the SVM.

\subsection{Study of HbA1c Control for Diabetic Patients}\label{RealData}
A fundamental aspect of managing diabetes is achieving tight control of blood sugar levels to reduce the risk of diabetes complications. Current diabetes guidelines recognize that treatment recommendations regarding tight control of HbA1c may not be appropriate for older patients, especially those with comorbid conditions, due to increased risk for adverse events \citep{american2018standards}. For example, complex patients may be at risk for negative outcomes when treated aggressively \citep{ismail2011individualizing, riddle2012individualizing}. These increased risks have even more significance if the benefits of tight control are limited due to shortened life expectancy. If the long-term benefits are not relevant and the short-term risks are significant, there may be little reason to pursue tight control of HbA1c. As a proof-of-concept analysis, we would like to recommend one-sided PDI for the upper bound of HbA1c such that the number hospitalization events are minimized. 

We use a de-identified electronic health record (EHR) dataset from a large health care system which contains laboratory records, information on primary care visits, and medication history of diabetes patients between 01/2003 to 12/2011. We use $45$ relevant features and $8126$ subjects whose observed HbA1c measures that are between $4\%$ and $10\%$ to learn the PDI of HbA1c. Details about the data set can be found in the supplement. For the purpose of notation consistency, we take the negative value of number of hospitalization events as the outcome. Hence a larger value of $Y$ represents a preferred outcome. By this coding, we also set $S=-0.5$ and $\alpha = 0.5$, which leads to the interpretation of PDI such that if a patient controls his/her HbA1c less than the upper bound, they would likely not experience diabetes-related hospitalization (i.e., with the probability of hospitalization less than $0.5$) in the next 90-days after HbA1c is measured. Since the observed outcome is a count, any choice of $S$ that is between $-1$ and $0$ will lead to the same result and interpretation as setting $S=-0.5$. 

We randomly select $1000$ patients as a training data set to learn the policy and use the remaining $7126$ patients as a testing data set. We repeat this random splitting $100$ times and report the averaged results. We apply all methods compared in our simulations to this real data. The tuning parameters are selected by using the same 5-fold cross-validation strategy.

To demonstrate the result, we plotted the distributions of the observed HbA1c level  and the estimated upper bounds of samples in one testing data set  in Figure ~\ref{Distribution}. While the majority of patients have observed HbA1c between 5.5\% and 7.5\%, their estimated upper bounds of the L-O-Learning and G-O-Learning are generally between 6.5\% and 8\%, which suggests that some patients may be over-controlling their HbA1c levels. In comparison, the PDI recommendations from the comparison methods are highly tilted towards the upper limit, especially the SVM. It is unlikely that physicians will recommend such high HbA1c values to patients. Furthermore, the grid search range of HbA1c levels for the indirect methods is between 4\% and 10\%, and indirect methods recommended a high proportion of patients to have 10\% as the HbA1c levels upper bound (logistic regression: 30\% of all patients; SVM: 70\%; RF: 50\%). If we change the grid search range to be between 4\% and 9\%,  the indirect methods would recommend a high proportion of patients to have 9\% as the upper bound. This indicates the policies estimated from indirect methods can have unwanted dependence on the grid search range. Such patterns of the estimated upper bounds of HbA1c levels shown in Figure ~\ref{Distribution} are similar across all $100$ replications.

\begin{figure*}[ht]
	\centering
	\includegraphics[width=12cm, height=7cm]{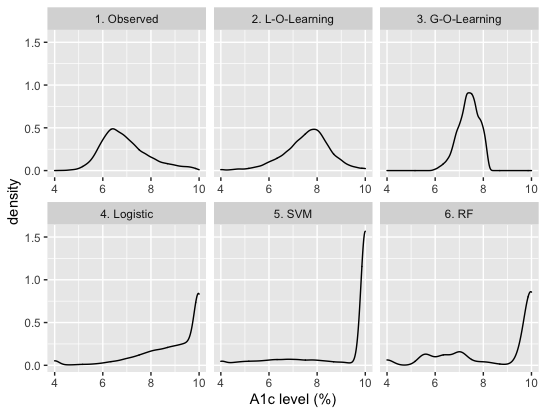}
	\caption{Density plot of observed and estimated upper bounds of HbA1c PDIs.} \label{Distribution}  
\end{figure*}

Furthermore, the upper bounds learned from our methods (especially the G-O-Learning) have much better agreement with the current medical guidance compared to other methods. In particular, the 2018 American Diabetes Association (ADA) Standards of Medical Care in Diabetes advises the following upper HbA1c bounds for diabetes patients \citep{american2018standards}: 6.5\% for people who can achieve this goal without experiencing a lot of hypoglycemia episodes or other negative effects of having lower blood glucose levels; 7.5\% for all children with diabetes; 8\% for people with a history of severe hypoglycemia, people who have had diabetes for many years or have difficulty achieving tighter control, and people with limited remaining life expectancy. In our dataset, 88\% of the patients are over 60 years old, 50\% over 73 years old, and 20\% over 80 years old. Most of these patients have diabetes for many years. Hence, the distributions of upper bounds recommendations from our methods coincide well with the standards given the age distribution. 

\section{CONCLUSION}\label{Conclusion}
We have proposed a framework for learning individualized dose intervals with desired theoretical properties. Our method extends a previous method that can only learn the policy for an optimal dosage \citep{chen2016personalized}. We proposed a DC algorithm to efficiently solve the underlying non-convex optimization problems with linear or nonlinear policies learning. 

Note that both direct and indirect methods proposed by us can learn PDI, meanwhile, we advocate using the direct method since it can recommend dose intervals without the need of fitting complex regression models. Learning PDI relies on the function restriction due to Conditions C1 and C2  in Section~\ref{Theory}, and such a restriction applies to both direct and indirect methods. If the conditions are not satisfied, the PDI may not be well defined or it will be a union of several discontinuous intervals. 

We also believe that $S$ should be chosen to have good interpretation. In our HbA1c example, choosing $S$ between $-1$ to $0$  leads to the same interpretation of the resulting PDI. For the applications where there is no obvious choice of $S$, we could set it to be a number between $\alpha \times 100\%$-th percentile of $Y(a_L)|X$ and  $\alpha \times 100\%$-th percentile of $Y(a_U)|X$ to ensure the existence of PDI.

Lastly, learning causally meaningful PDIs using observed data relies on the no-unmeasured confounder assumption. In our HbA1c example, we worked closely with our physician collaborator to include curated variables from EHR and medical claims to ensure this assumption. At the same time, sensitivity analysis may help examine the robustness of the PDI policy. We want to emphasize that before implementing an estimated policy for PDI in practice, it is critically important to verify the superiority of the estimated policy over existing clinical practice via randomized trials. 

\subsection*{Acknowledgments}
We would like to thank all reviewers for their helpful comments that greatly improved this paper.
Research reported in this work was partially funded through a Patient-Centered Outcomes Research
Institute (PCORI) Award (ME-2018C2-13180). The views in this work are solely the responsibility
of the authors and do not necessarily represent the views of the Patient-Centered
Outcomes Research Institute (PCORI), its Board of Governors or Methodology Committee.

\bibliography{doseinterval}

\begin{thebibliography}{}

\bibitem[{American Association of Clinical Endocrinologists},
  2019]{american2019management}
{American Association of Clinical Endocrinologists} (2019).
\newblock Management of common comorbidities of diabetes.

\bibitem[{American Diabetes Association}, 2018]{american2018standards}
{American Diabetes Association} (2018).
\newblock Standards of medical care in diabetes—2018 abridged for primary
  care providers.
\newblock {\em Clinical diabetes: a publication of the American Diabetes
  Association}, 36(1):14.

\bibitem[Austin, 2018]{austin2018assessing}
Austin, P.~C. (2018).
\newblock Assessing the performance of the generalized propensity score for
  estimating the effect of quantitative or continuous exposures on binary
  outcomes.
\newblock {\em Statistics in Medicine}, 37(11):1874--1894.

\bibitem[Bica et~al., 2020]{bica2020estimating}
Bica, I., Jordon, J., and van~der Schaar, M. (2020).
\newblock Estimating the effects of continuous-valued interventions using
  generative adversarial networks.
\newblock {\em arXiv preprint arXiv:2002.12326}.

\bibitem[Borghans et~al., 2012]{borghans2012length}
Borghans, I., Kleefstra, S.~M., Kool, R.~B., and Westert, G.~P. (2012).
\newblock Is the length of stay in hospital correlated with patient
  satisfaction?
\newblock {\em International Journal for Quality in Health Care},
  24(5):443--451.

\bibitem[Breiman, 2001]{breiman2001random}
Breiman, L. (2001).
\newblock Random forests.
\newblock {\em Machine learning}, 45(1):5--32.

\bibitem[Breiman et~al., 1984]{breiman1984classification}
Breiman, L., Friedman, J., Stone, C.~J., and Olshen, R.~A. (1984).
\newblock {\em Classification and regression trees}.
\newblock CRC press.

\bibitem[Chakraborty and Moodie, 2013]{chakraborty2013statistical}
Chakraborty, B. and Moodie, E. (2013).
\newblock {\em Statistical methods for dynamic treatment regimes}.
\newblock Springer.

\bibitem[Chen et~al., 2016]{chen2016personalized}
Chen, G., Zeng, D., and Kosorok, M.~R. (2016).
\newblock Personalized dose finding using outcome weighted learning.
\newblock {\em Journal of the American Statistical Association},
  111(516):1509--1521.

\bibitem[Chen et~al., 2017]{Chen2017}
Chen, S., Tian, L., Cai, T., and Yu, M. (2017).
\newblock A general statistical framework for subgroup identification and
  comparative treatment scoring.
\newblock {\em Biometrics}, 73(4):1199--1209.

\bibitem[Collins and Varmus, 2015]{collins2015new}
Collins, F.~S. and Varmus, H. (2015).
\newblock A new initiative on precision medicine.
\newblock {\em New England journal of medicine}, 372(9):793--795.

\bibitem[Eberts et~al., 2013]{eberts2013optimal}
Eberts, M., Steinwart, I., et~al. (2013).
\newblock Optimal regression rates for svms using gaussian kernels.
\newblock {\em Electronic Journal of Statistics}, 7:1--42.

\bibitem[Flack and Adekola, 2020]{flack2020blood}
Flack, J.~M. and Adekola, B. (2020).
\newblock Blood pressure and the new acc/aha hypertension guidelines.
\newblock {\em Trends in cardiovascular medicine}, 30(3):160--164.

\bibitem[Fong et~al., 2018]{fong2018covariate}
Fong, C., Hazlett, C., Imai, K., et~al. (2018).
\newblock Covariate balancing propensity score for a continuous treatment:
  Application to the efficacy of political advertisements.
\newblock {\em The Annals of Applied Statistics}, 12(1):156--177.

\bibitem[Freedman et~al., 2010]{freedman2010comparison}
Freedman, B.~I., Shenoy, R.~N., Planer, J.~A., Clay, K.~D., Shihabi, Z.~K.,
  Burkart, J.~M., Cardona, C.~Y., Andries, L., Peacock, T.~P., Sabio, H.,
  et~al. (2010).
\newblock Comparison of glycated albumin and hemoglobin a1c concentrations in
  diabetic subjects on peritoneal and hemodialysis.
\newblock {\em Peritoneal Dialysis International}, 30(1):72--79.

\bibitem[Huling et~al., 2021]{huling2021independence}
Huling, J.~D., Greifer, N., and Chen, G. (2021).
\newblock Independence weights for causal inference with continuous exposures.
\newblock {\em arXiv preprint arXiv:2107.07086}.

\bibitem[Imbens and Rubin, 2015]{imbens2015causal}
Imbens, G.~W. and Rubin, D.~B. (2015).
\newblock {\em Causal inference in statistics, social, and biomedical
  sciences}.
\newblock Cambridge University Press.

\bibitem[{International Warfarin Pharmacogenetics Consortium},
  2009]{international2009estimation}
{International Warfarin Pharmacogenetics Consortium} (2009).
\newblock Estimation of the warfarin dose with clinical and pharmacogenetic
  data.
\newblock {\em New England Journal of Medicine}, 360(8):753--764.

\bibitem[Ismail-Beigi et~al., 2011]{ismail2011individualizing}
Ismail-Beigi, F., Moghissi, E., Tiktin, M., Hirsch, I.~B., Inzucchi, S.~E., and
  Genuth, S. (2011).
\newblock Individualizing glycemic targets in type 2 diabetes mellitus:
  implications of recent clinical trials.
\newblock {\em Annals of internal medicine}, 154(8):554--559.

\bibitem[Kallus and Santacatterina, 2019]{kallus2019kernel}
Kallus, N. and Santacatterina, M. (2019).
\newblock Kernel optimal orthogonality weighting: A balancing approach to
  estimating effects of continuous treatments.
\newblock {\em arXiv preprint arXiv:1910.11972}.

\bibitem[Kallus and Zhou, 2018]{kallus2018policy}
Kallus, N. and Zhou, A. (2018).
\newblock Policy evaluation and optimization with continuous treatments.
\newblock In {\em International Conference on Artificial Intelligence and
  Statistics}, pages 1243--1251.

\bibitem[Kosorok and Laber, 2019]{kosorok2019precision}
Kosorok, M.~R. and Laber, E.~B. (2019).
\newblock Precision medicine.
\newblock {\em Annual Review of Statistics and Its Application}, 6:263--286.

\bibitem[Lanckriet and Sriperumbudur, 2009]{lanckriet2009convergence}
Lanckriet, G.~R. and Sriperumbudur, B.~K. (2009).
\newblock On the convergence of the concave-convex procedure.
\newblock In {\em Advances in Neural Information Processing Systems}, pages
  1759--1767.

\bibitem[Le~Thi~Hoai and Tao, 1997]{le1997solving}
Le~Thi~Hoai, A. and Tao, P.~D. (1997).
\newblock Solving a class of linearly constrained indefinite quadratic problems
  by {DC} algorithms.
\newblock {\em Journal of global optimization}, 11(3):253--285.

\bibitem[Liang and Yu, 2020]{Liang2020}
Liang, M. and Yu, M. (2020).
\newblock A semiparametric approach to model effect modification.
\newblock {\em Journal of the American Statistical Association}, 0(0):1--13.

\bibitem[Little et~al., 2011]{little2011status}
Little, R.~R., Rohlfing, C.~L., and Sacks, D.~B. (2011).
\newblock Status of hemoglobin a1c measurement and goals for improvement: from
  chaos to order for improving diabetes care.
\newblock {\em Clinical chemistry}, 57(2):205--214.

\bibitem[Lou et~al., 2018]{Lou2017}
Lou, Z., Shao, J., and Yu, M. (2018).
\newblock Optimal treatment assignment to maximize expected outcome with
  multiple treatments.
\newblock {\em Biometrics}, pages 506--516.

\bibitem[Luedtke and van~der Laan, 2016]{alexander2016Comment}
Luedtke, A.~R. and van~der Laan, M.~J. (2016).
\newblock Comment on: `personalized dose finding using outcome weighted
  learning'.
\newblock {\em Journal of the American Statistical Association},
  111(516):1526--1530.

\bibitem[Mancinelli et~al., 2000]{mancinelli2000pharmacogenomics}
Mancinelli, L., Cronin, M., and Sad{\'e}e, W. (2000).
\newblock Pharmacogenomics: the promise of personalized medicine.
\newblock {\em Aaps Pharmsci}, 2(1):29--41.

\bibitem[Moodie et~al., 2012]{moodie2012q}
Moodie, E.~E., Chakraborty, B., and Kramer, M.~S. (2012).
\newblock Q-learning for estimating optimal dynamic treatment rules from
  observational data.
\newblock {\em Canadian Journal of Statistics}, 40(4):629--645.

\bibitem[Moodie et~al., 2009]{moodie2009estimating}
Moodie, E.~E., Platt, R.~W., and Kramer, M.~S. (2009).
\newblock Estimating response-maximized decision rules with applications to
  breastfeeding.
\newblock {\em Journal of the American Statistical Association},
  104(485):155--165.

\bibitem[Nitin, 2010]{nitin2010hba1c}
Nitin, S. (2010).
\newblock Hba1c and factors other than diabetes mellitus affecting it.
\newblock {\em Singapore Med J}, 51(8):616--622.

\bibitem[Radin, 2014]{radin2014pitfalls}
Radin, M.~S. (2014).
\newblock Pitfalls in hemoglobin a1c measurement: when results may be
  misleading.
\newblock {\em Journal of general internal medicine}, 29(2):388--394.

\bibitem[Rich et~al., 2014]{rich2014simulating}
Rich, B., Moodie, E.~E., and Stephens, D.~A. (2014).
\newblock Simulating sequential multiple assignment randomized trials to
  generate optimal personalized warfarin dosing strategies.
\newblock {\em Clinical trials}, 11(4):435--444.

\bibitem[Riddle and Karl, 2012]{riddle2012individualizing}
Riddle, M.~C. and Karl, D.~M. (2012).
\newblock Individualizing targets and tactics for high-risk patients with type
  2 diabetes: practical lessons from accord and other cardiovascular trials.
\newblock {\em Diabetes Care}, 35(10):2100--2107.

\bibitem[Schulz and Moodie, 2020]{schulz2020doubly}
Schulz, J. and Moodie, E.~E. (2020).
\newblock Doubly robust estimation of optimal dosing strategies.
\newblock {\em Journal of the American Statistical Association}, pages 1--13.

\bibitem[Scott et~al., 2017]{scott2017genome}
Scott, J.~G., Berglund, A., Schell, M.~J., Mihaylov, I., Fulp, W.~J., Yue, B.,
  Welsh, E., Caudell, J.~J., Ahmed, K., Strom, T.~S., et~al. (2017).
\newblock A genome-based model for adjusting radiotherapy dose (gard): a
  retrospective, cohort-based study.
\newblock {\em The lancet oncology}, 18(2):202--211.

\bibitem[Shawe-Taylor et~al., 2004]{shawe2004kernel}
Shawe-Taylor, J., Cristianini, N., et~al. (2004).
\newblock {\em Kernel methods for pattern analysis}.
\newblock Cambridge university press.

\bibitem[Sondhi et~al., 2020]{sondhi2020balanced}
Sondhi, A., Arbour, D., and Dimmery, D. (2020).
\newblock Balanced off-policy evaluation in general action spaces.
\newblock In {\em International Conference on Artificial Intelligence and
  Statistics}, pages 2413--2423. PMLR.

\bibitem[Steinwart and Christmann, 2008]{steinwart2008support}
Steinwart, I. and Christmann, A. (2008).
\newblock {\em Support vector machines}.
\newblock Springer Science \& Business Media.

\bibitem[Tibshirani, 1996]{tibshirani1996regression}
Tibshirani, R. (1996).
\newblock Regression shrinkage and selection via the lasso.
\newblock {\em Journal of the Royal Statistical Society. Series B
  (Methodological)}, pages 267--288.

\bibitem[T{\"u}bbicke, 2020]{tubbicke2020entropy}
T{\"u}bbicke, S. (2020).
\newblock Entropy balancing for continuous treatments.
\newblock {\em arXiv preprint arXiv:2001.06281}.

\bibitem[Vapnik, 2013]{vapnik2013nature}
Vapnik, V. (2013).
\newblock {\em The nature of statistical learning theory}.
\newblock Springer science \& business media.

\bibitem[Wu and Liu, 2007]{wu2007robust}
Wu, Y. and Liu, Y. (2007).
\newblock Robust truncated hinge loss support vector machines.
\newblock {\em Journal of the American Statistical Association},
  102(479):974--983.

\bibitem[Xu et~al., 2015]{Xu2015}
Xu, Y., Yu, M., Zhao, Y.-Q., Li, Q., Wang, S., and Shao, J. (2015).
\newblock Regularized outcome weighted subgroup identification for differential
  treatment effects.
\newblock {\em Biometrics}, 71(3):645--653.

\bibitem[Zhao et~al., 2009]{zhao2009reinforcement}
Zhao, Y., Kosorok, M.~R., and Zeng, D. (2009).
\newblock Reinforcement learning design for cancer clinical trials.
\newblock {\em Statistics in medicine}, 28(26):3294--3315.

\bibitem[Zhou et~al., 2021]{zhou2021parsimonious}
Zhou, W., Zhu, R., and Zeng, D. (2021).
\newblock A parsimonious personalized dose-finding model via dimension
  reduction.
\newblock {\em Biometrika}, 108(3):643--659.

\bibitem[Zhu et~al., 2020]{zhu2020kernel}
Zhu, L., Lu, W., Kosorok, M.~R., and Song, R. (2020).
\newblock Kernel assisted learning for personalized dose finding.
\newblock In {\em Proceedings of the 26th ACM SIGKDD International Conference
  on Knowledge Discovery and Data Mining}, KDD '20, page 56–65.

\end{thebibliography}
\bibliographystyle{apalike}


\clearpage
\appendix

\thispagestyle{empty}

\onecolumn \makesupplementtitle

\section{DC  ALGORITHM FOR ONE-SIDED PDI}
Let $f_L(X_i; \mathbf{v}) = \sum_{j=1}^{n}v_j k(X_i, X_j) + v_0$ and $\mathbf{v} = (v_1,.....,v_n, v_0)$.
As claimed in Section 3.2 in the main paper,  the DC algorithm  repeatedly updates 
\begin{equation}
	\begin{split}
		\mathbf{v}^{t+1}=\arg\min_\mathbf{v} \left(\mathcal{L}_1(\mathbf{v})-[\nabla \mathcal{L}_2(\mathbf{v}^t)]^\top(\mathbf{v}-\mathbf{v}^t)\right).
	\end{split}
\end{equation}
We derive the form of the quadratic programming by first showing that 
\begin{equation}
	\begin{split}
		\nabla \mathcal{L}_2(\mathbf{v})&={\nabla}_\mathbf{v}\Bigg(\sum_{i=1}^n \frac{1}{P(A_i|X_i)} \Big[\big(1-\alpha\big)I(Y_i>S)\Psi_{\epsilon, 2}\big(f_L(X_i), A_i\big) + \alpha I(Y_i \leq S) \Psi_{2,\epsilon}\big(A_i, f_L(X_i) \big)\Big]\Bigg)\nonumber\\
		&=\sum_{i=1}^n\frac{1}{P(A_i|X_i)} \Big[ \big(1-\alpha\big) I(Y_i>S) Q_i(\beta)
		- \alpha I(Y_i \leq S) \tilde Q_i(\beta)\Big] \nabla_\mathbf{v}f_L(X_i) \nonumber
	\end{split}
\end{equation}
where $\nabla_\mathbf{v}f_L(X_i) = (k(X_i, X_1), \ldots, k(X_i, X_n), 1)$, $Q_i(\mathbf{v})=I\big(f_L(X_i; \mathbf{v}) - A_i >  \epsilon \big)/\epsilon$ and  $\tilde Q_i(\mathbf{v})=I\big(A_i- f_L(X_i; \mathbf{v})>\epsilon\big)/\epsilon$  \\

Following that, we have 
\begin{align}
	[\nabla \mathcal{L}_2(\mathbf{v}^{t})]^\top \mathbf{v} &=
	\sum_{i=1}^n\frac{1}{P(A_i|X_i)}\Big[ \big(1-\alpha\big)I(Y_i > S)Q_i(\mathbf{v}^{t})
	- \alpha I(Y_i \leq S)\tilde Q_i(\mathbf{v}^{t}) \Big] \big\{\sum_{j=1}^{n}v_jk(X_i,X_j) + v_0 \big\} \nonumber 
\end{align}

By plugging in the above results into the DC procedure, 
\begin{align}
	\mathbf{v}^{t+1}&=\arg\min_\mathbf{v} \hskip0.02\textwidth \mathcal{L}_1(\mathbf{v})-[\nabla \mathcal{L}_2(\mathbf{v}^t)]^\top(\mathbf{v} -\mathbf{v}^t) = \arg\min_\beta\hskip0.02\textwidth \mathcal{L}_1(\mathbf{v})-[\nabla \mathcal{L}_2(\mathbf{v}^t)]^\top \mathbf{v} \nonumber \\
	&=  \arg\min_\mathbf{v}  \hskip0.02\textwidth \sum_{i=1}^n\frac{1}{P(A_i|X_i)} \Bigg(\big(1-\alpha\big)I(Y_i > S)\Big[\Psi_{1,\epsilon}\big(f_L(X_i; \mathbf{v}) ,A_i\big)- Q_i(\mathbf{v}^t) (\sum_{j=1}^{n}v_jk(X_i,X_j) + v_0) \Big]\nonumber \\
	& \hskip0.02\textwidth + \alpha I(Y_i \leq S)\Big[\Psi_{1,\epsilon}\big(A_i, f_L(X_i; \mathbf{v}) \big) +  \tilde Q_i(\mathbf{v}^t) (\sum_{j=1}^{n}v_jk(X_i,X_j) + v_0) \Big] \Bigg)  +\frac{\lambda_n}{2} \sum_{i=1}^{n}\sum_{j=1}^{n}v_i v_j k(X_i, X_j)   \label{fulltarget-app}
\end{align}

To solve the above optimization problem, we now show that it is essentially a quadratic programming problem which can be easily solved using standard optimization packages. 

Let $H_i=\frac{(1-\alpha)I(Y_i>S)}{P(A_i|X_i)\lambda_n}$ and $\tilde H_i=\frac{\alpha (Y_i \leq S)}{P(A_i|X_i)\lambda_n}$. Minimizing Equation \eqref{fulltarget-app} is equivalent to 
\begin{align*}
	\arg\min_{\mathbf{v}} & \sum_{i=1}^{n}H_i \xi_i+ \sum_{i=1}^{n} \tilde H_i \tilde \xi_i +  \sum_{i=1}^{n}(\tilde H_i\tilde Q_i^t - H_i Q_i^t)(\sum_{j=1}^{n}v_jk(X_i,X_j) + v_0) +\frac{1}{2} \sum_{i=1}^{n}\sum_{j=1}^{n}v_i v_j k(X_i, X_j)  \\
	s.t.\;\;\; & \xi_i - [ (\sum_{j=1}^{n}v_jk(X_i,X_j) + v_0) - A_i]/\epsilon \geq 0,\hskip0.1\textwidth \xi_i\geq 0\\
	&\tilde \xi_i  - [A_i - ( \sum_{j=1}^{n}v_j k(X_i, X_j) + v_0)]/\epsilon \geq 0,\hskip0.1\textwidth \tilde \xi_i\geq 0
\end{align*}

Applying the Lagrangian multiplier, the problem above is equivalent to 
\begin{align}
	& \max_{\alpha_i,\tilde \alpha_i,\mu_i,\tilde \mu_i} \min_{\mathbf{v},\xi}\sum_{i=1}^{n}H_i \xi_i+ \sum_{i=1}^{n} \tilde H_i \tilde \xi_i + \sum_{i=1}^{n}(\tilde H_i\tilde Q_i^t - H_i Q_i^t)(\sum_{j=1}^{n}v_jk(X_i,X_j) +v_0) +\frac{1}{2}\sum_{i=1}^{n}\sum_{j=1}^{n}v_i v_j k(X_i, X_j) \nonumber\\
	&-\sum_{i=1}^{n} \alpha_i( \xi_i - [(\sum_{j=1}^{n}v_jk(X_i,X_j)  + v_0) - A_i]/\epsilon)-\sum_{i=1}^{n} \tilde \alpha_i(\tilde \xi_i  - [A_i - (\sum_{j=1}^{n}v_jk(X_i,X_j)  + v_0)]/\epsilon )-\sum_{i=1}^{n} \mu_i\xi_i-\sum_{i=1}^{n}\tilde \mu_i\tilde \xi_i\nonumber\\
	s.t.& \hskip0.05\textwidth \sum_{i=1}^{n}v_i k(X_i, X_j) +\sum_{i=1}^{n} (\tilde H_i\tilde Q_i^t - H_i Q_i^t +\alpha_i/\epsilon - \tilde \alpha_i/\epsilon) k(X_i,X_j)  =0 \;\;\; \forall j=1,\dots,n \nonumber\\
	&\hskip0.05\textwidth \sum_{i=1}^{n} (\tilde H_i\tilde Q_i^t - H_i Q_i^t +\alpha_i/\epsilon - \tilde \alpha_i/\epsilon) =0\nonumber\\
	&\hskip0.05\textwidth H_i - \alpha_i-\mu_i=0\;\;\; \forall i=1,\dots,n\nonumber\\
	&\hskip0.05\textwidth \tilde H_i - \tilde \alpha_i-\tilde \mu_i=0\;\;\; \forall i=1,\dots,n\label{dual-single}
\end{align}

Incorporating the constraint that $\alpha_i\geq0,\;\tilde \alpha_i\geq0,\;\mu_i\geq0,\;\tilde \mu_i\geq0$ and denoting $\beta_i = \tilde H_i\tilde Q_i^t - H_i Q_i^t +\alpha_i /\epsilon - \tilde \alpha_i /\epsilon$, we have the final quadratic programming problem as 
\begin{align}
	\min_{\beta} \hskip0.05\textwidth& \frac{1}{2}\sum_{i=1}^{n}\sum_{j=1}^{n}  \beta_i k(X_i,X_j) \beta_j + \sum_{i=1}^{n} \beta_i A_i \nonumber\\
	s.t.\hskip0.05\textwidth &  \tilde H_i\tilde Q_i^t - H_i Q_i^t - \tilde H_i/\epsilon \leq \beta_i\leq \tilde H_i\tilde Q_i^t - H_i Q_i^t + H_i/\epsilon \;\;\; \forall i=1,\dots,n \nonumber\\ 
	& \sum_{i=1}^{n} \beta_i =0  \label{quadratic-single-app}
\end{align}

The $t+1$ step solution of $(v_1, \ldots, v_n)$ is thus attained from 
\begin{align}
	\hskip0.05\textwidth \sum_{i=1}^{n}v_i k(X_i, X)  = -\sum_{i=1}^{n} \beta_i k(X_i,X)   \nonumber 
\end{align}

As for the intercept $v_0$, although there exists an analytic solution by applying the KKT conditions, we solve it using a line search for simplicity. We iteratively solve the quadratic programming problem until convergence.

\section{PROOF OF THEOREMS FOR ONE-SIDED PDI}
\subsection{Proof of Theorem \ref{THMconsist-p-1}}\label{thmconsist-p-1}
\begin{proof}
	First, let $f^*\in F\equiv\{f:\forall x\in \mathcal{X}, P(Y(f(x)) >S)=\alpha\}$. For any $f'$ such that $f'\neq f^*$, denote $\mathcal{X}_0=\{x:x\in\mathcal{X},\exists f\in F\; s.t.\;f'(x)=f(x)\}$, $\mathcal{X}_1=\{x:x\in\mathcal{X},\;f'(x)>\max_{f\in F}f(x)\}$, and $\mathcal{X}_2=\{x:x\in\mathcal{X},\;f'(x)<\min_{f\in F}f(x)\}$\\
	
	Under the ignorability assumption, $P(Y>S |A=f'(x),x)=P(Y(f'(x))>S|x)>\alpha$ for $x\in\mathcal{X}_1$, while $P(Y>S|A=f'(x),x)=P(Y(f'(x))>S|x)<\alpha$ for $x\in\mathcal{X}_2$. 
	
	Then
	\begin{align*}
		&\hspace{-1cm} R(f') - R(f^*) \\
		=&E\Big[\frac{1}{P(A\,|\,X)}\bigg(\big(1-\alpha\big)I\big(Y>S\big)\Big(I\big(A<f'(X)\big)-I\big(A<f^*(X)\big)\Big)\\
		&+\alpha I\big(Y \leq S\big)\Big(I\big(A>f'(X)\big)-I\big(A>f^*(X)\big)\Big)\bigg)\Big] \\
		=&E\Bigg[\int_{f^*(X)}^{f'(X)}\Big(\big(1-\alpha\big)P\big(Y>S|A=a,X\big)\\
		&\hspace{0.15\textheight}-\alpha P\big(Y \leq S|A=a,X\big)\Big)da\bigg|X\in \mathcal{X}_1\Bigg]P(X\in \mathcal{X}_1) \tag{A1}\\
		&+E\Bigg[\int_{f'(X)}^{f^*(X)}\Big(-\big(1-\alpha\big)P\big(Y>S|A=a,X\big)\\
		&\hspace{0.15\textheight}+\alpha P\big(Y \leq S|A=a,X\big)\Big)da\bigg|X\in \mathcal{X}_2\Bigg]P\big(X\in \mathcal{X}_2\big)\tag{A2}\\
		&+E\Bigg[\int_{f^*(X)}^{f'(X)}\Big(\big(1-\alpha\big)\alpha-\alpha\big (1-\alpha\big)\Big)da\bigg|X\in \mathcal{X}_0\Bigg]P\big(X\in \mathcal{X}_0\big)\tag{A3}\\
	\end{align*}
	
	Here 
	\[
	\begin{split}
		(A3)=&0\\
		(A1)=&E\left[\int_{f^*(X)}^{f'(X)}\Big(P\big(Y>S|A=a,X\big)-\alpha\Big)da\bigg|X\in \mathcal{X}_1\right]P\big(X\in \mathcal{X}_1\big)\nonumber \\
		=&E\left[\int_{\max_{f\in\mathcal{F}} f(X)}^{f'(X)}\Big(P\big(Y>S|A=a,X\big)-\alpha\Big)da\bigg|X\in \mathcal{X}_1\right]P\big(X\in \mathcal{X}_1\big)\nonumber \\
		\geq& 0\\
		>&0 \qquad \text{  if  }\quad P(X\in \mathcal{X}_1)\neq 0\nonumber \\
	\end{split}
	\]
	
	\[
	\begin{split}
		(A2)=&E\left[\int_{f'(X)}^{f^*(X)}\Big(-P\big(Y>S|A=a,X\big)+\alpha\Big)da\bigg|X\in \mathcal{X}_2\right]P(X\in \mathcal{X}_2)\nonumber \\
		=&E\left[\int_{f'(X)}^{\min_{f\in\mathcal{F}}f(X))}\Big(-P\big(Y>S|A=a,X\big)+\alpha\Big)da\bigg|X\in \mathcal{X}_2\right]P(X\in \mathcal{X}_2)\nonumber \\
		\geq& 0\\
		>&0 \qquad \text{  if  }\quad P(X\in \mathcal{X}_2)\neq 0\nonumber \\
	\end{split}
	\]
	
	Therefore, $R(f')-R(f^*)\geq 0$. And $R(f')-R(f^*)=0$ if and only if $P(X\in\mathcal{X}_1)=P(X\in\mathcal{X}_2)=0$, i.e, $f'\in F$. Hence any possible optimal left bound function $f_{L,opt}(x)= \arg\min_f R(f)$ belongs to $F$. In addition, according to the assumptions (C1) and (C2) of  partial monotonicity for PDI, we have that $P\left(Y(f_{L,opt}(x) )>S\right)=\alpha$ and $P\left(Y(a'(x))>S\right)\geq \alpha$ for $\forall x\in\mathcal{X}$, where  $a'(x)$ is an arbitrary measurable function s.t. $f_{L,opt}(x)\leq a'(x)\leq a_U$. Theorem \ref{THMconsist-p-1} is proved.
	
\end{proof}

\subsection{Proof of Theorem \ref{THMbound-p-1}}\label{thmbound-p-1}
\begin{proof}
	For any measurable function $f:\mathcal{X}\rightarrow \mathbb{R}$, we have 
	\begin{align*} 
		&\hspace{-1cm} |R(f)-R_{\epsilon_n}(f)|\\
		=&E\bigg|\frac{1}{\epsilon_nP(A\,|\,X)}\bigg(\big(1-\alpha\big)I(Y>S)\max(\epsilon_n-(f(X)-A)_+,0)\\
		&+\alpha I(Y \leq S)\max(\epsilon_n-(A-f(X))_+,0)\bigg)\bigg|\\
		=&E_X\bigg|\frac{1}{\epsilon_n}\int_{a\in [a_L,a_U]}\bigg(\big(1-\alpha\big)P(Y>S|A=a,X)\max(\epsilon_n-(f(X)-a)_+,0)\\
		&\qquad\qquad +\alpha P(Y \leq S|A=a,X)\max(\epsilon_n-(a-f(X))_+,0)\bigg)da\;\bigg|\\
		=&E_X\bigg|\frac{\big(1-\alpha\big)}{\epsilon_n}\int_{a\in [f(X)-\epsilon_n,f(X)]}P(Y>S|A=a,X)(\epsilon_n-(f(X)-a))da\\
		&\qquad\qquad +\frac{\alpha}{\epsilon_n}\int_{a\in [f(X),f(X)+\epsilon_n]} P(Y \leq S|A=a,X)(\epsilon_n-(a-f(X)))da\;\bigg|\\
		=&E_X\bigg|\frac{\big(1-\alpha\big)}{\epsilon_n}\int_{z\in[0,\epsilon_n]} P(Y>S|A=f(X)-z,X)(\epsilon_n-z)dz\\
		&\qquad\qquad +\frac{\alpha}{\epsilon_n}\int_{z\in[0,\epsilon_n]}  P(Y \leq S|A=f(X)+z,X)(\epsilon_n-z)dz\;\bigg|\\
		&\leq \epsilon_n
	\end{align*}
	Therefore Theorem \ref{THMbound-p-1} follows. 
\end{proof}

\subsection{Proof of Theorem \ref{THMoptimalrate-1}}\label{thmoptimalrate-1}
\begin{proof}

	According to Theorem \ref{THMbound-p-1}, we have 
	\begin{align}
		R(\hat f_L)-R( f_{opt})&\leq R_{\epsilon_n}(\hat f_L)-R_{\epsilon_n}( f_{L,opt})+2C\epsilon_n\nonumber\\
		&\leq R_{\epsilon_n}(\hat f_L)-R_{\epsilon_n}( f_{L}^*)+2C\epsilon_n\nonumber\\
		&\leq \bigg(\lambda_n||\hat f_L||_k^2+R_{\epsilon_n}(\hat f_L)-R_{\epsilon_n}^*\bigg)+  \bigg(2C\epsilon_n   \bigg) \nonumber\\
		&= (I)+(II)\label{totalupper}
	\end{align}
	
	where $f_{L}^*$ is the minimizer of $R_{\epsilon_n}(\cdot)$ and $R_{\epsilon_n}^*=R_{\epsilon_n}(f_{L}^*)$. \\
	
	To bound  \eqref{totalupper}, we need to first bound $(I)$. As in \cite{chen2016personalized}, we refer to Theorem 7.23 in \cite{steinwart2008support}. In order to use the oracle inequality in the theorem, there are four conditions to be satisfied.
	\begin{itemize}
		\item (B1) The loss function $L(\cdot)$ has a supremum bound $L(\cdot)\leq B$ for a constant $B>0$.
		\item (B2) The loss function $L(\cdot)$ is locally Lipschitz continuous and can be clipped at a constant $M>0$ such that $\tilde f=I(|f|\leq M)f+I(|f|\geq M)M$.
		\item (B3) The variance bound $E_P\big(L\circ \tilde{f} -L\circ f^*_{L,P}\big)^2\leq V\Big(E_P\big(L\circ \tilde{f} -L\circ f^*_{L,P}\big)\Big)^v$ is satisfied for  a constant $v\in[0,1]$, $V\geq B^{2-v}$ and all $f\in \mathcal{H}$.
		\item (B4) For fixed $n\geq 1$, there exist constants $p\in (0,1)$ and $a\geq B$ such that the entropy number $E_{D_X\sim P_X^n} e_i(id:\mathcal{H}\rightarrow L_2(D_X))\leq ai^{-\frac{1}{2p}}$, $i\leq 1$.
	\end{itemize} 
		
	To verify Condition (B1), recall that the relaxed loss function of the lower bound of the one-sided PDI is 
	\begin{align*}
		&L_{\epsilon}(X,A,f_L(X))=\frac{1}{P(A\,|\,X)}\Big\{\big(1-\alpha\big)I(Y > S) \Psi_\epsilon({f}_L(X),A) +\alpha I( Y \leq S) \Psi_\epsilon(A,{f}_L(X))\Big\}
	\end{align*}
	
	Assuming the inverse probability ${1}/{P(A\,|\,X)}$ is bounded by a constant $B$, then  $L_{\epsilon}(X,A,f_L(X))$ is naturally bounded by $B$ as well.
	
	To verify Condition (B2), similar to (B1), $L_{\epsilon}(X,A,f_L(X))$ is Lipschitz continuous with a Lipschitz constant $B$. Also, $f_L$ can be clipped to have a smaller risk. Suppose that $M$ is an upper bound of the absolute value of a reasonable range of dose, say, $M=\max\{|a_L|,|a_U|\}$, and $\tilde f_L=I(|f_L|\leq M)f_L+I(|f_L|\geq M)M$. It naturally follows that $R(\tilde f_L)\leq R(f_L)$, since any unreasonably large dose recommendation introduces a larger risk. 
	
	Condition (B3) is satisfied with $v=0$ and $V=4B^2$,  because  
	\begin{align}
		E\big[\big(L_{\epsilon}\circ \tilde f_L -L_{\epsilon}\circ f_L^*\big)^2\big]\leq 2E\big[\big(L_{\epsilon}\circ \tilde f_L\big)^2 +\big(L_{\epsilon}\circ f_L^*\big)^2\big]\leq 4B^2\nonumber
	\end{align}
	
	The Gaussian Kernel used in Sections \ref{Simulation} and  \ref{RealData} is one type of benign kernels. According to Theorem 7.34 of \cite{steinwart2008support}, (B4) is satisfied with the constant $a=c_{\epsilon,p}\gamma_n^{-\frac{(1-p)(1+\zeta )d}{2p}}$, where $d/(d+\tau)<p<1$ and $\zeta >0$ are two constants. 
	
	Since Conditions (B1)-(B4) are satisfied in our case, applying Theorem 7.23 from \cite{steinwart2008support} yields 
	\begin{align}
		(I)&\leq 9*\bigg \{\lambda_n|| f_L^0||_k^2+R_{\epsilon_n}(f_L^0)-R_{\epsilon_n}^*\bigg\}+K_0\left[\frac{1}{\gamma_n^{(1-p)(1+\zeta)d}\lambda_n^pn}\right]^{\frac{1}{2-p}}+36\sqrt{2}B\Big(\frac{\tau}{n}\Big)^{\frac{1}{2}}+\frac{15B\tau}{n}\label{approximation error}
	\end{align}
	with probability $1-3e^{-\tau}$.
	
	In order to bound $\big(\lambda_n|| f_L^0||_k^2+R_{\epsilon_n}(f_L^0)-R_{\epsilon_n}^*\big)$, the approximation error, we refer to  Section 2 in \cite{eberts2013optimal}. As \eqref{approximation error} holds for any $f_L^0\in \mathcal{H}_\gamma$, we construct a specific $f_L^0$ to facilitate the proof. First, similar to  Equation (8) in \cite{eberts2013optimal}, we define a function $\mathcal{K}(x)=\sum_{j=1}^{r} {r\choose j}(-1)^{1-j}\frac{1}{j^d}(\frac{2}{\gamma^2 })^{d/2}\mathcal{K}_{j\gamma/\sqrt{2}}(x)$ where $\mathcal{K}_\gamma=\exp(-\gamma^2||x||_2^2)$ for all $x\in R^d$, and subsequently define $f_L^0$ via convolution as follows,
	\begin{align*}
		f_L^0=\mathcal{K}*f_L^*=\int_{R^d}\mathcal{K}(X-t)f_L^*(t)dt,\;\; x\in R^d
	\end{align*}
	
	If we assume $f^*_L\in L_2(R^d)\cap L_\infty(R^d)$, then from Theorem 2.3 in  \cite{eberts2013optimal}, it can be shown that $f_L^0\in \mathcal{H}_\gamma$, where $\mathcal{H}_\gamma$ is the RKHS with the Gaussian kernel. In addition, Theorem 2.2 provides the upper bound for the access risk of $f_L^0$, which can be incorporated to yield the following results. 
	\begin{align}
		&\lambda_n|| f_L^0||_k^2+R_{\epsilon_n}(f_L^0)-R_{\epsilon_n}^*\nonumber\\
		&=\lambda_n|| \mathcal{K}*f_L^*||_k^2+R_{\epsilon_n}(\mathcal{K}*f_L^*)-R_{\epsilon_n}^*\nonumber\\
		&\leq \lambda_n(\gamma_n\sqrt{\pi})^{-d}(2^r-1)^2||f_L^*||^2_{L_2(R^d)}+R_{\epsilon_n}(\mathcal{K}*f_L^*)-R_{\epsilon_n}^*\nonumber\\
		&\leq \lambda_n(\gamma_n\sqrt{\pi})^{-d}(2^r-1)^2||f_L^*||^2_{L_2(R^d)}+B|\mathcal{K}*f_L^*-f_L^*|_{L_1(P_X)}\nonumber\\
		&\leq \lambda_n(\gamma_n\sqrt{\pi})^{-d}(2^r-1)^2||f_L^*||^2_{L_2(R^d)}+BC_{r,1}||g||_{L_p(P_X)}\omega_{r,L_1(R^d)}(f^*_L,\gamma_n/2)\label{access risk} 
	\end{align}
	
	Given the fact that $f_{opt}\in B_{1,\infty}^\delta(R^d)$, if we further assume $f_{L}^*\in B_{1,\infty}^\delta(R^d)$, i.e., $B^\delta_{1,\infty}(R^d)=\{f\in L_\infty (R^d):\sup_{t>0} (t^{-\delta}\omega_{r,L_1(R^d)})(f,t)<\infty\}$, then $w_{r,L_1}(R^d)(f_L^*,\gamma_n/2)<c_0\gamma^\delta$, where $c_0$ is a constant. Plugging this into \eqref{access risk}, we obtain
	\begin{align*}
		\lambda_n|| f_L^0||_k^2+R_{\epsilon_n}(f_L^0)-R_{\epsilon_n}^*\leq c_1\lambda_n\gamma_n^{-d}+c_2\gamma_n^\delta
	\end{align*}
	
	After combining all the parts together, 
	\begin{align*}
		R(\hat f_L)-R( f_{opt})\leq c_1\frac{\lambda_n}{\gamma_n^d}+c_2\gamma_n^\delta+c_3\frac{1}{\gamma_n^{\frac{(1-p)(1+\zeta)d}{2-p}}\lambda_n^\frac{p}{2-p}n^\frac{1}{2-p}}+c_4\frac{\tau^{\frac{1}{2}}}{n^{\frac{1}{2}}}+c_5\frac{\tau}{n}+c_6\epsilon_n
	\end{align*}
	By properly choosing the constant, i.e.
	\begin{align*}
		\gamma_n\propto \bigg(\frac{1}{n}\bigg)^{\frac{1}{2\delta+d}}\hspace{0.05\textheight}\lambda_n\propto \bigg(\frac{1}{n}\bigg)^{\frac{\delta+d}{2\delta+d}}\hspace{0.05\textheight}\epsilon_n=\mathcal{O}\bigg(n^{-\frac{\delta}{2\delta+d}}\bigg)
	\end{align*}
	we have the convergence rate:
	\begin{align*}
		R(\hat f_L)-R( f_{opt})=\mathcal{O}\bigg(n^{-\frac{\delta}{2\delta+d}}\bigg)
	\end{align*}
\end{proof}

\section{THEORETICAL RESULTS FOR TWO-SIDED PDI}
To ensure the two-sided PDI exists, we need the following assumptions similar to the one-sided PDI case:

(C3) For any $x$, the boundaries $a_L$ and $a_U$ for the dose range satisfy $P\{Y(a_L)>S \big |x\}\leq \alpha$ and $P\big\{Y(a_U)>S |x \big\}\leq  \alpha$.\\ [8pt]
(C4) There exists a function  $a^{*}(\cdot) \in[a_L,a_U]$ such that $P\{Y(a^*(x))>S|x\}\geq  \alpha$ \\ [8pt]
(C5) For any $x$, there are two numbers $a_{1,x} < a_{2,x}$ with $a_{1,x} \in [a_L,a_U]$ and $a_{2,x} \in [a_L,a_U]$ such that  $P\{Y(a_{1,x})>S |x\} = \alpha$ and  $P\{Y(a_{2,x})>S |x\} = \alpha$.

\begin{corollary}\label{THMconsist-p-2}
	Let $\left[ f_{L,opt}, f_{U,opt} \right]=\arg\min_{f_L\leq f_U} R(f_L,f_U)$. Then, under the assumptions (C3)-(C5), $\left[f_{L,opt},f_{U,opt}\right]$ has the following properties:\\ [8pt]
	(1) $P\big\{Y\big(f_{L,opt}(x)\big)> S \, |\, x \big\}= \alpha \;  \text{and} \; P\big\{Y\big(f_{U,opt}(x)\big)> S \, |\, x \big\}= \alpha, \; \forall x$. \\  [8pt]
	(2) For any measurable function $a(\cdot)$ s.t. $f_{L,opt}(x)\leq a(x)\leq  f_{U,opt}(x)$,  $\forall x\in \mathcal{X}$, the potential outcome $Y(a(x))$ satisfies $P(Y(a(x)) > S|x)\geq \alpha$. \\ [8pt]
	(3)  For any $a_1(x) <f_{L,opt}(x)$, $P\big(Y(a_1(x))>S|x \big) < \alpha$,  and for any $a_2(x) >f_{U,opt}(x)$, $P\big(Y(a_2(x))>S|x\big) < \alpha$.
\end{corollary}

\begin{proof}
	First, $$F_L\equiv\{f:\forall x\in \mathcal{X}, P(Y(f(x))>S|x)=\alpha\quad and\quad \exists a^*<f(x),\;P(Y(a^*(x))>S|x)\leq \alpha\}$$ $$F_U\equiv\{f:\forall x\in \mathcal{X}, P(Y(f(x))>S|x)=\alpha\quad and\quad \exists a^*>f(x),\;P(Y(a^*(x))>S|x)\leq \alpha\}$$
	
	Let $f^*_L\in F_L$ and $f^*_U\in F_U$, then for arbitrary $f'_L$ and $f'_U$, the sample space of $\mathcal{X}$ can be decomposed as $\mathcal{X}=\mathcal{X}_0\cup \mathcal{X}_1\cup \mathcal{X}_2\cup \mathcal{X}_3\cup \mathcal{X}_4 \cup \mathcal{X}_5 \cup \mathcal{X}_6$
	where 
	\begin{align*}
		&\mathcal{X}_0=\{x:x\in\mathcal{X},\quad\exists f_L\in F_L\; s.t.\;f'_L(x)=f_L(x)\quad and \quad \exists f_U\in F_U\; s.t.\;f'_U(x)=f_U(x)\}\\
		&\mathcal{X}_1=\{x:x\in\mathcal{X},\quad \max_{f_L\in F_L} f_L(x)< f'_L(x)\leq   f'_U(x)<\min_{f_U\in F_U} f_U(x) \}\\
		&\mathcal{X}_2=\{x:x\in\mathcal{X},\quad f'_L(x)<\min_{f_L\in F_L} f_L(x)\leq  \max_{f_U\in F_U} f_U(x)<f'_U(x) \}\\
		&\mathcal{X}_3=\{x:x\in\mathcal{X},\quad f'_L(x)\leq \min_{f_L\in F_L} f_L(x)\leq f'_U(x)\leq  \max_{f_U\in F_U} f_U(x) \}\\
		&\mathcal{X}_4=\{x:x\in\mathcal{X},\quad  \min_{f_L\in F_L} f_L(x)\leq f'_L(x)\leq  \max_{f_U\in F_U} f_U(x)\leq f'_U(x) \}\\
		&\mathcal{X}_5=\{x:x\in\mathcal{X},\quad f'_L(x) \leq f'_U(x)<\min_{f_L\in F_L} f_L(x)\leq  \max_{f_U\in F_U} f_U(x) \}\\
		&\mathcal{X}_6=\{x:x\in\mathcal{X},\quad  \min_{f_L\in F_L} f_L(x)\leq   \max_{f_U\in F_U} f_U(x)<f'_L(x)\leq f'_U(x) \}
	\end{align*}
	Then
	\begin{align*}
		R(f')-R(f^*) =&E\bigg[\frac{1}{P(A\,|\,X)}\bigg(\big(1-\alpha\big)I\big(Y>S\big)\Big(I(A\notin [f'_L(X),f'_U(X)])-I(A\notin [f^*_L(X),f^*_U(X)])\Big)\\
		&+\alpha I\big(Y \leq S\big)\Big(I(A\in [f'_L(X),f'_U(X)])-I(A\in [f^*_L(X),f^*_U(X)])\Big)\bigg)\bigg] \\
		= &E\left[\dots\big|X\in \mathcal{X}_1\right]P(X\in \mathcal{X}_0) \tag{M0}\\
		 &+E\left[\dots\big|X\in \mathcal{X}_1\right]P(X\in \mathcal{X}_1) \tag{M1}\\
		&+E\left[\dots\big|X\in \mathcal{X}_2\right]P(X\in \mathcal{X}_2) \tag{M2}\\
		&+E\left[\dots\big|X\in \mathcal{X}_3\right]P(X\in \mathcal{X}_3) \tag{M3}\\
		&+E\left[\dots\big|X\in \mathcal{X}_4\right]P(X\in \mathcal{X}_4) \tag{M4}\\
		&+E\left[\dots\big|X\in \mathcal{X}_5\right]P(X\in \mathcal{X}_5) \tag{M5}\\
		&+E\left[\dots\big|X\in \mathcal{X}_6\right]P(X\in \mathcal{X}_6) \tag{M6}    
	\end{align*}
	where  $(M3)\geq 0$, $(M3)>0$ if $P(X\in \mathcal{X}_3)>0$, and $(M4)\geq 0$, $(M4)>0$ if $P(X\in \mathcal{X}_3)>0$, according to the results of Theorem \ref{THMconsist-p-1}. These results can be directly applied since it is easily seen that within $\mathcal{X}_3$ and $\mathcal{X}_4$,  assumptions for Theorem \ref{THMconsist-p-1} are satisfied. Obviously $(M0)=0$ and
	\begin{align*}    
		(M1)=&E\Bigg[\int_{f^*_L(X)}^{f'_L(X)}\Big(\big(1-\alpha\big)P(Y>S|A=a,X)-\alpha P(Y \leq S|A=a,X)\Big)da\\
		&+\int^{f^*_U(X)}_{f'_U(X)}\Big(\big(1-\alpha\big)P(Y>S|A=a,X)-\alpha P(Y \leq S|A=a,X)\Big)da\bigg|X\in \mathcal{X}_1\Bigg]P(X\in \mathcal{X}_1) \\
		=&E\bigg[\int_{\max_{f_L\in F_L}f_L(X)}^{f'_L(X)}\Big(P(Y>S|A=a,X)-\alpha \Big)da\\
		&+\int^{\min_{f_U\in F_U}f_U(X)}_{f'_U(X)}\Big(P(Y>S|A=a,X)-\alpha \Big)da\bigg|X\in \mathcal{X}_1\bigg]P(X\in \mathcal{X}_1) 
\end{align*}
Therefore $(M1)\geq 0$ and $>0$  if $P(X\in \mathcal{X}_1)\neq 0$. Similarly,	\begin{align*}    
		(M2)=&E\bigg[\int^{f^*_L(X)}_{f'_L(X)}\Big(-\big(1-\alpha\big)P(Y>S|A=a,X)+\alpha P(Y \leq S|A=a,X)\Big)da\\
		&+\int_{f^*_U(X)}^{f'_U(X)}\Big(-\big(1-\alpha\big)P(Y>S|A=a,X)+\alpha P(Y \leq S|A=a,X)\Big)da\bigg|X\in \mathcal{X}_2\bigg]P(X\in \mathcal{X}_2) \\
		=&E\bigg[\int^{\min_{f_L\in F_L}f_L(X)}_{f'_L(X)}\Big(-P(Y>S|A=a,X)+\alpha \Big)da\\
		&+\int_{\max_{f_U\in F_U}f_U(X)}^{f'_U(X)}\Big(-P(Y>S|A=a,X)+\alpha \Big)da\bigg|X\in \mathcal{X}_2\bigg]P(X\in \mathcal{X}_2) \end{align*}
Therefore	$(M2)\geq 0$ and $>0$  if $P(X\in \mathcal{X}_2)\neq 0$. 
	\begin{align*}
		(M5)=&E\Bigg[\int_{f'_L(X)}^{f'_U(X)}\Big(-\big(1-\alpha\big)P(Y>S|A=a,X)+\alpha P(Y \leq S|A=a,X)\Big)da\\
		&+\int^{f^*_L(X)}_{f^*_U(X)}\Big(\big(1-\alpha\big)P(Y>S|A=a,X)-\alpha P(Y \leq S|A=a,X)\Big)da\bigg|X\in \mathcal{X}_5\Bigg]P(X\in \mathcal{X}_5) \\
		=&E\bigg[\int_{f'_L(X)}^{f'_U(X)}\Big(-P(Y>S|A=a,X)+\alpha \Big)da\\
		&+\int^{f^*_L(X)}_{f^*_U(X)} \Big(P(Y>S|A=a,X)-\alpha \Big)da\bigg|X\in \mathcal{X}_5\bigg]P(X\in \mathcal{X}_5) 
	\end{align*}
	Therefore	$(M5)\geq 0$ and $>0$  if $P(X\in \mathcal{X}_2)\neq 0$. 
	
	\begin{align*}
		(M6)=&E\Bigg[\int_{f'_L(X)}^{f'_U(X)}\Big(-\big(1-\alpha\big)P(Y>S|A=a,X)+\alpha P(Y \leq S|A=a,X)\Big)da\\
		&+\int^{f^*_L(X)}_{f^*_U(X)}\Big(\big(1-\alpha\big)P(Y>S|A=a,X)\\
		&\hspace{0.15\textheight}-\alpha P(Y \leq S|A=a,X)\Big)da\bigg|X\in \mathcal{X}_6\Bigg]P(X\in \mathcal{X}_6) \\
		=&E\bigg[\int_{f'_L(X)}^{f'_U(X)}\Big(-P(Y>S|A=a,X)+\alpha \Big)da\\
		&+\int^{f^*_L(X)}_{f^*_U(X)} \Big(P(Y>S|A=a,X)-\alpha \Big)da\bigg|X\in \mathcal{X}_6\bigg]P(X\in \mathcal{X}_6) 
	\end{align*}
		Therefore	$(M6)\geq 0$ and $>0$  if $P(X\in \mathcal{X}_6)\neq 0$. Consequently, $R(f'_U,f'_U) - R(f^*_L,f^*_U)\geq 0$, and $R(f'_U,f'_U)-R(f^*_L,f^*_U)=0$ if and only if $P(X\notin \mathcal{X}_0)=0$. Hence for any optimal two-sided bound functions $[f_{L,opt},f_{U,opt}]=\arg \min_{f_L\leq f_U} R(f_L,f_U)$, we have  $f_{L,opt}\in F_L$ and $f_{L,opt}\in F_U$. According to the assumptions (C3), (C4) and (C5), we have that $P\left(Y\big(f_{L,opt}(x)\big)>S|x \right)=P\left(Y\big(f_{U,opt}(x)\big)>S| x\right)=\alpha$ and $P\left(Y(a'_{x} )>S| x\right)\geq \alpha$ for $\forall x\in\mathcal{X}$, where  $a'(x)$ is an arbitrary measurable function s.t.  $f_{L,opt}(x)\leq a'(x)\leq f_{U,opt}(x)$. 
\end{proof}

\begin{corollary}\label{THMbound-p-2}
	For any interval of two measurable functions  $\left[f_L(x),f_U(x)\right]$, where  $f_L(x)\leq f_U(x), \forall x\in\mathcal{X} $, we have $|R(f_L,f_U)-R_{\epsilon_n}(f_L,f_U)|\leq C\epsilon_n$, where $C$ is a constant. 
\end{corollary}

\begin{proof}
	For any measurable function $f_L:\mathcal{X}\rightarrow \mathbb{R}$, $f_U:\mathcal{X}\rightarrow \mathbb{R}$ and $f_L(X)\leq f_U(X)\forall x\in \mathcal{X}$, we have 
	\begin{align*} 
		|&R(f)-R_{\epsilon_n}(f)|\\
		&=E_X\bigg|\frac{\big(1-\alpha\big)}{\epsilon_n}\int_{a\in [f_L(X)-\epsilon_n,f_L(X)]}P(Y>S|A=a,X)(\epsilon_n-(f_L(X)-a))da\\
		&\qquad\qquad +\frac{\alpha}{\epsilon_n}\int_{a\in [f_L(X),f_L(X)+\epsilon_n]} P(Y \leq S|A=a,X)(\epsilon_n-(a-f_L(X)))da\\
		&\qquad\qquad +\frac{\alpha}{\epsilon_n}\int_{a\in [f_U(X)-\epsilon_n,f_U(X)]} P(Y \leq S|A=a,X)(\epsilon_n-(f_U(X)-a))da\\
		&\qquad\qquad +\frac{\big(1-\alpha\big)}{\epsilon_n}\int_{a\in [f_U(X),f(X)_R+\epsilon_n]}P(Y>S|A=a,X)(\epsilon_n-(a-f_U(X)))da\;\bigg|\\
		&=E_X\bigg|\frac{\big(1-\alpha\big)}{\epsilon_n}\int_{z\in[0,\epsilon_n]} P(Y>S|A=f_L(X)-z,X)(\epsilon_n-z)dz\\
		&\qquad\qquad +\frac{\alpha}{\epsilon_n}\int_{z\in[0,\epsilon_n]}  P(Y \leq S|A=f_L(X)+z,X)(\epsilon_n-z)dz\\
		&\qquad\qquad +\frac{\alpha}{\epsilon_n}\int_{z\in[0,\epsilon_n]}  P(Y \leq S|A=f_U(X)-z,X)(\epsilon_n-z)dz\\
		&\qquad\qquad +\frac{\big(1-\alpha\big)}{\epsilon_n}\int_{z\in[0,\epsilon_n]} P(Y>S|A=f_U(X)+z,X)(\epsilon_n-z)dz\; \bigg|\\
		&=2 \epsilon_n. 
	\end{align*}
	
\end{proof}

We provide a convergence rate by borrowing the results from Theorem \ref{THMoptimalrate-1}, and replace the convergence of the two-sided PDI estimator with two estimators of one-sided PDIs. In order to do so, extra assumptions have to be made. It would be an interesting future work if better rates under weaker assumptions can be achieved.
\begin{corollary}\label{THMoptimalrate-2}
	Assume that $f_{ L,opt}\in B^{\delta}_{1,\infty}(\mathbb{R}^d)$ and $f_{ U,opt}\in B^{\delta}_{1,\infty}(\mathbb{R}^d)$ where $B^{\delta}_{1,\infty}(\mathbb{R}^d)=\{f\in L_{\infty}((\mathbb{R}^d)):\sup_{t>0} (t^{-\delta}w_{(r,L_1)}(f,t))<\infty\}$ and $w$ is the modulus of continuity. In addition, assume that there exists a measurable function $f_{M,opt}(x)$, which is known or can be estimated consistently, such that $f_{L,opt}(x)< f_{M}(x)<f_{U,opt}(x)$ for all $x\in \mathcal{X}$. Then for any $\eta>0$, $d/(d+\tau)<p<1$, $\tau>0$, $0<c<1$,and properly chosen $\gamma_n$, $\lambda_n$ and $\epsilon_n$, we have 
	\begin{equation} 
		R(\hat f_{L,n},\hat f_{U,n})-R(f_{L,opt},f_{U,opt})=\mathcal{O}_p\Bigg(n^{-\frac{\delta}{2\delta+d}}\Bigg)
	\end{equation}
	with probability $c \big(1-3e^{-\tau}\big)^2$. Here $d$ is the dimension of $\mathcal{X}$.
\end{corollary}

\begin{proof}
	
	The proof of Corollary \ref{THMoptimalrate-2} consists of two steps. In the first step, we assume that there exists a function or a consistent estimator of a function that separates the dataset into two subsets which we call as left and right subsets, in each of which Conditions (C1) and (C2) are satisfied. In the second step, the integrated excessive risk of the two-sided PDI estimator is bounded by the sum of excessive risks of two independent one-sided PDI estimators.
	
	Suppose there is a function $f_{M,opt}(x)$, such that $f_{L,opt}(x)<f_{M,opt}(x)<f_{U,opt}(x)$ for all $x\in \mathcal{X}$. Such a function is typically unknown in practice, hence we consider the case when $f_{M,opt}(x)$ is estimated from data. One choice of $f_{M,opt}$ is $f_{\max,opt}=\arg\max_{f}E\big[Y|A=f(X),X\big]$, the optimal individual dose rule (IDR) in \cite{chen2016personalized}. Theoretical properties of the estimator $\hat f_{\max}$ have been established in  \cite{chen2016personalized}  and also discussed in \cite{alexander2016Comment}. Despite the fact that it is suggested in \cite{alexander2016Comment} that the O-learning based estimator might convergence in nearly $\mathcal{O}(n^{-1/2})$, the rate established there was based on the risk (or value function) instead of the rule $f_{\max,opt}$ itself. To facilitate proving the properties of the two-sided PDI, we make the following further assumption for the estimated optimal dose function $\hat f_{\max}$.
	\begin{align}
		f_{L,opt}(x)< \hat f_{\max}(x)< f_{U,opt}(x),\quad\forall x\in \mathcal{X}\label{assumption-inside}
	\end{align}
	
	In the second step, two subsets are defined where the lower bound function and the upper bound function can be estimated separately using the corresponding subset, assuming there is a $ f_{M,opt}$, known or estimated, such that $f_{L,opt}(x)<  f_{M,opt}(x)< f_{U,opt}(x)$ for all $x\in \mathcal{X}$. Denote the original dataset as $\mathcal{D}=\{\big(X_i,A_i,Y_i\big),\;1\leq i\leq n\}$  and then we can subsequently define the left subset as $\mathcal{D_L}=\{\big(X_i,A_i,Y_i\big),\;\forall i\; s.t.\; A_i\leq  f_{M,opt}(X_i)\}$ and the right subset as $\mathcal{D_U}=\{\big(X_i,A_i,Y_i\big),\;\forall i\; s.t.\; A_i>  f_{M,opt}(X_i)\}$. These two subsets are independent as there is no shared observations. 
	
	Now define $\hat f_L'=\min_{f_L} \hat R_{\epsilon}(f_L)$ on  $\mathcal{D_L}$ and $\hat f_U'=\min_{f_U} \hat R_{\epsilon}(f_U)$ on  $\mathcal{D_U}$, while $\big\{\hat f_L,\hat f_U\big\}=\min_{f_L,f_U}\hat R_{\epsilon}(f_L,f_U)$ on the dataset of $\mathcal{D}$. Similarly, define $ f_{L,*}'=\min_{f_L} R_{\epsilon}(f_L)$ on the (sub)population of $\mathcal{D_L}$ and $ f_{U,*}'=\min_{f_U} R_{\epsilon}(f_U)$ on the (sub)population of $\mathcal{D_U}$, while $\big\{f_{L,*},f_{U,*}\big\}=\min_{f_L,f_U}R_{\epsilon}(f_L,f_U)$ on the original population.  To simplify the problem, we assume that $f_{L,*}(x)-f_{U,*}(x)\leq -2\epsilon$, since this inequality will hold when $\epsilon\rightarrow 0$ as $n\rightarrow 0$. By the definitions of $\mathcal{D_L}$, $\mathcal{D_U}$ and $\mathcal{D}$, we have $\big\{f_{L,*}',f_{U,*}'\big\}=\big\{f_{L,*},f_{U,*}\big\}$, i.e., the population estimates of the lower and upper bounds estimated jointly  are the same as the population estimates of the lower and upper bounds estimated individually while assuming the other part is known. In addition, we make the following assumption of the uniform convergence of the empirical risk function. 
	
	\begin{align}
		\sup_{f_L,f_U\in\mathcal{H}}\big|\hat R_{\epsilon}(f_L,f_U)-R_{\epsilon}(f_L,f_U)\big|\leq \kappa\sqrt{\frac{\log (n)}{n}}\;\text{with probability }c \label{VC inequality}
	\end{align}
	where $\kappa$ is a constant related to $c$ and the complexity of the RKHS.

	We will show how the loss function of the two-sided PDI can be decomposed as the sum of two one-sided PDI losses plus a quantity not related to the dose. Recall that the loss function for the lower bound of the one-sided PDI is
	\begin{align}
		&L_{\epsilon}(X,A,f_L(X))=\frac{1}{P(A\,|\,X)}\Big\{\big(1-\alpha\big)I(Y > S) )\Psi_\epsilon({f}_L(X),A)+\alpha I( Y \leq S) )\Psi_\epsilon(A,{f}_L(X))\Big\}\nonumber
	\end{align}
	
	The loss function for the upper bound of the one-sided PDI is
	\begin{align}
		&L_{\epsilon}(X,A,f_U(X))=\frac{1}{P(A\,|\,X)}\Big\{\big(1-\alpha\big)I(Y > S) )\Psi_\epsilon(A,{f}_U(X)) +\alpha I( Y \leq S) )\Psi_\epsilon({f}_U(X),A)\Big\}.\nonumber
	\end{align}
	
	And the loss function for the two-sided PDI is
	\begin{align}
		&L_{\epsilon}(X,A,f_L(X),f_U(X))=\frac{1}{P(A\,|\,X)}\Big\{\big(1-\alpha\big)I(Y > S) )\Psi^{out}_\epsilon({f}_L(X),A,f_U(X))\nonumber\\
		&\hspace{0.3\linewidth}+\alpha I( Y \leq S) )\Psi^{in}_\epsilon(f_L(X),A,{f}_U(X))\Big\}.\nonumber
	\end{align}
	
	It follows that the difference between the two-sided loss and the sum of the one-sided losses is a constant which does not contain the bound functions, under the assumption that $f_{L,*}(x)-f_{U,*}(x)\leq -2\epsilon$. 
	\begin{align}
		&L_{\epsilon}(X,A,f_L(X),f_U(X))-L_{\epsilon}(X,A,f_L(X))-L_{\epsilon}(X,A,f_U(X))\nonumber\\
		&=\frac{1}{P(A\,|\,X)}\Big\{\big(1-\alpha\big)I(Y > S) \Psi_\epsilon(A,f_U(X)) +\alpha I( Y \leq S)\big(\Psi_\epsilon(f_U(X),A)-1\big)\Big\}-L_{\epsilon}(X,A,f_U(X))\nonumber\\
		&=\frac{-1}{P(A\,|\,X)}\Big\{\alpha I( Y \leq S)\Big\}\nonumber
	\end{align}
	
	Combined with Corollary \ref{THMbound-p-2}, we have  
	\begin{align}
		R(\hat f_L,&\hat f_U)-R( f_{L,opt},f_{U,opt})\leq R_{\epsilon_n}(\hat f_L,\hat f_U)-R_{\epsilon_n}( f_{L,opt},f_{U,opt})+2C\epsilon_n\nonumber\\
		&\leq R_{\epsilon_n}(\hat f_L,\hat f_U)-R_{\epsilon_n}( f_{L,*},f_{U,*})+2C\epsilon_n\nonumber\\
		&\leq R_{\epsilon_n}(\hat f_L',\hat f_U')-R_{\epsilon_n}( f_{L,*},f_{U,*})+2C\epsilon_n+2\kappa\sqrt{\frac{\log(n)}{n}}\nonumber\\
		&= R_{\epsilon_n}(\hat f_L')-R_{\epsilon_n}( f_{L,*}')+R_{\epsilon_n}(\hat f_U')-R_{\epsilon_n}( f_{U,*}')+2C\epsilon_n+2\kappa\sqrt{\frac{\log(n)}{n}}\label{two-sided rate final}
	\end{align}
	
	with probability $c$ as $n\rightarrow \infty$. The probability $c$ comes from the inequality \eqref{VC inequality}. The first inequality of \eqref{two-sided rate final} is guaranteed by Theorem \ref{THMconsist-p-2}. The second inequality is from the definition of $f_{L,*}$ and $f_{U,*}$. The third inequality is from the assumption \eqref{VC inequality}. The coefficient is $2\times \kappa$ as the uniform bound has to be applied twice because the definitions of $\hat f_L$ and $\hat f_U$ only guarantee $\hat R_{\epsilon_n}(\hat f_L,\hat f_U)<\hat R_{\epsilon_n}(\hat f_L',\hat f_U')$  but not  $R_{\epsilon_n}(\hat f_L,\hat f_U)<R_{\epsilon_n}(\hat f_L',\hat f_U')$.  The last equality is from the decomposition of the two-sided loss and the equivalence of $\big\{f_{L,*}',f_{U,*}'\big\}$ and $\big\{f_{L,*},f_{U,*}\big\}$. 
	
	The theorem follows by applying \eqref{approximation error} twice, because $\hat f_L'$ and $\hat f_U'$ are independently estimated on two non-overlapping datasets $\mathcal{D_L}$ and $\mathcal{D_U}$. That is, $R_{\epsilon_n}(\hat f_L')-R_{\epsilon_n}( f_{L,*}')$ is independent with $R_{\epsilon_n}(\hat f_U')-R_{\epsilon_n}( f_{U,*}')$. The uniform bound \eqref{VC inequality} and Theorem \ref{THMoptimalrate-1} hold simultaneously with a probability of at least  $c(1-3e^{-\tau})^2$. Here $\epsilon_n$ only shows up in the term $2C\epsilon_n$. Hence by choosing $\epsilon_n=\mathcal{O}\big(n^{-\frac{\delta}{2\delta+d}}\big)$, it will not impact the convergence rate. The term $2\kappa\sqrt{\frac{\log(n)}{n}}$ corresponds to the convergence rate of $\mathcal{O}\big(\sqrt{\frac{\log n}{n}}\big)$. Hence, \eqref{two-sided rate final} is dominated either by the first two risk terms or by $2\kappa\sqrt{\frac{\log(n)}{n}}$.
\end{proof}

\section{ADDITIONAL SIMULATIONS WITH NO CONFOUNDING COVARIATES}

In this section, we consider scenarios similar to those in the main paper except that there is no confounding. In particular, the treatment $A$ of the training data  follows a uniform distribution  $\text{Unif}(-2,2)$ (i.e., independent of $X$) while other data generating mechanisms are the same as in the main paper. We repeat each setting for $100$ times and report results by evaluating the empirical risks on corresponding independent testing data sets with sample sizes of $10000$.  Similar to the settings with confounding, our methods perform better than the competitive methods. 

\begin{table*}[!htbp]
	\centering
	\caption{Empirical risk $R(\hat{f}_L)$ (with SD) under settings without confounding covariates} \label{simresults_supp}
	\begin{tabular}{clllllll}
		\hline
		& n  & d & L-O-Learning & G-O-Learning& Logistic & SVM & RF \\ 
		\hline
		Scenario 1 & 200 & 10  & 0.137 (0.008)  & \textbf{0.132 (0.011) } & 0.154 (0.024) & 0.139 (0.011) & 0.136 (0.008) \\ 
		$\sigma^2 = 9$				& 200 & 50  &  0.157 (0.014) & \textbf{0.133(0.011) }& 0.185(0.023) & 0.181 (0.043) & 0.135 (0.012) \\ 
		& 400 & 10  & 0.128 (0.006)  & \textbf{0.126 (0.006)} & 0.132(0.011) & 0.129 (0.006) & 0.130 (0.009) \\ 
		& 400 & 50  &  0.154 (0.006) & \textbf{0.128 (0.007) }& 0.174 (0.022) & 0.142 (0.008) & 0.129 (0.008) \\ 
		\hline
		Scenario 2  & 200 & 10 & 0.136 (0.008) & \textbf{0.130 (0.011}) & 0.153 (0.022) & 0.138 (0.011) & 0.133 (0.012) \\ 
		$\sigma^2 = 9$   	& 200 & 50  & 0.159 (0.013)  & \textbf{0.134 (0.010) }& 0.194 (0.024) & 0.226 (0.037) & 0.139 (0.017) \\ 
		& 400 & 10  & 0.129 (0.005)  & \textbf{0.123 (0.006) }& 0.132 (0.012) & 0.127 (0.005) & 0.124 (0.005) \\ 
		& 400 & 50  & 0.156 (0.007)  & \textbf{0.126 (0.008) }& 0.180 (0.026) & 0.144 (0.008) & 0.127 (0.004) \\ 
		\hline
	\end{tabular}
\end{table*}

\section{MORE BACKGROUND ABOUT THE HBA1C CONTROL STUDY}\label{RealDataSupp}

Hemoglobin is an iron-containing oxygen which transports protein in the red blood cells. HbA1c, the most abundant minor hemoglobin in the human body, is formed when glucose accumulates in red blood cells and binds to the hemoglobin. This process occurs slowly and continuously over the lifespan of red blood cells, which is 120 days on average. This makes HbA1c an ideal biomarker of long-term glycemic control. Patients who are susceptible to high blood glucose are usually recommended to have their HbA1c levels regularly measured and recorded. 

HbA1c test, since became commercially available in 1978, has been one of the standard tools for monitoring diabetes progression. American Diabetes Association (ADA) recommended using A1c measurement in 1988. The Diabetes Control and Complications Trial (DCCT) demonstrated its importance as a predictor of diabetes-related outcomes, and the ADA started recommending specific A1c targets in 1994 \citep{little2011status}. In 2010, ADA added HbA1c$\geq 6.5\%$ as a criterion for diabetes diagnosis. The prevalence of HbA1c test can be partially explained by its clinic convenience. As HbA1c is unaffected by acute perturbations in glucose levels, there is no need for fasting or timed samples. For people without diabetes, the normal range for the HbA1c level is between 4\% and 5.6\%. HbA1c levels between 5.7\% and 6.4\% imply a higher chance of getting diabetes. Levels of 6.5\% or higher indicate the presence of diabetes.

Despite the popularity of HbA1c tests, HbA1c may not reflect the real progression of diabetes given various conditions that patients are predisposed to. Interfering factors may lead to false results of HbA1c tests \citep{radin2014pitfalls}. In some cases, the direction of bias is predictable while in other cases not. For instance, HbA1c can be falsely elevated by iron deficiency anemias, vitamin B-12 anemias, folate deficiency anemias, asplenia and other conditions associated with decreased red cell turnover \citep{nitin2010hba1c}. HbA1c measures may be lower for patients with conditions that shorten the life of the red blood cell or increase its turnover rate  \citep{freedman2010comparison}, including acute and chronic blood loss, hemolytic anemia, splenomegaly and end-stage renal disease. While in some other cases, the direction of bias is not always easily predictable. For example, the complex interplay of glycemic control and other treatments, such as erythropoietin therapy and the treatment of uremia may influence HbA1c level in a more case-by-case way \citep{radin2014pitfalls}.

Due to the susceptibility of HbA1c measures to various medical and physiological factors, HbA1c measures reflect glycemia differently for different patients. Therefore, the recommendation of a single HbA1c control level for all the patients may not always be desirable.  A lower glycemic level is not always better, and over-controlling glycemia can lead to faint or compromised life quality. Especially for patients who are older, with severe complications, or in the phase of early post-surgery rehabilitation, HbA1c is not supposed to be controlled for them as low as for younger and healthier patients. 

As a response to the progress in medical research,  the idea of glucose control has been evolving over time. In the past, an HbA1c of 7\% was considered the golden rule of health for everyone. In recent years, however, the importance of a patient-centered approach to managing HbA1c levels has been recognized, which may better correspond to the patient's needs for diabetes management and their personal conditions and preferences. To fill this gap, we introduce the HbA1c dataset to study the relationship between HbA1c control and the events of diabetic patients. Specifically, the sample is a subset of a type II diabetes dataset which is drawn from a large Midwestern multi-specialty physician and patient group.  Electronic Health Record (EHR) datasets are merged with Medicare claim data to obtain laboratory records, information on primary care visits, and medication history. All patients have been tracked for a specific time period from the first quarter of 2003 to the fourth quarter of 2011. The full observational records are not available for the majority of patients.

Patients identified with diabetes must have at least 1 inpatient or SNF Medicare claim or at least two carrier claims for diabetes, which should be less than 2 years apart. The diagnosis is defined by one or more of the following categories: the ICD-9 codes, Diabetes mellitus (250.xx), Polyneuropathy in diabetes (357.2), Diabetes retinopathy (362.0x), and Diabetic cataract (366.41). There are several types of patients  excluded from the dataset, including those who are without continuous coverage with Medicare Parts A \& B for a baseline year or for at least one subsequent quarter, without Medicare railroad benefits or were not enrolled in a Medicare HMO. Each patient needs to have at least 5 quarter time periods in one of the 3 insurers from 2003 to 2011 including 4 baseline quarters and a measurement period of at least 1 subsequent quarter. Also, the patients should have A1c records over the 5 quarter period which are available in the health system. In the sample of analysis, 51.3\% of patients are female and 91.8\% of the patients are white. The ages of participants at the cohort start time range from 18 to 102 years old, with an average of 65.7 and a standard deviation of 13.7.

We regard the HbA1c measurement as the treatment dose in this study. Observations are centered around the availability of HbA1c measurements. Each observation is associated with a unique HbA1c measurement from an individual patient following the baseline period. The outcome variable is the number of negative events, which include hospitalizations and Emergency Department (ED) visits that cover 90 days after the HbA1c measurement. In addition, we also include a list of short-term covariates that we average within 90 days prior to the HbA1c measurement as well as a list of long-term covariates that are averaged from up to 900 days prior to the HbA1c measurement. After removing subjects with missing baseline covariates, $8126$ patients remain. An illustration of the data structure can be found in Figure~\ref{fig:realdata-structure}. Each observation is generated from an outcome period (90 days after the HbA1c measure), a short-term covariate period (90 days before the HbA1c measure) and a long-term covariate period (including the baseline information).

\begin{figure}[!htbp]
	\centering
	\includegraphics[width=15cm]{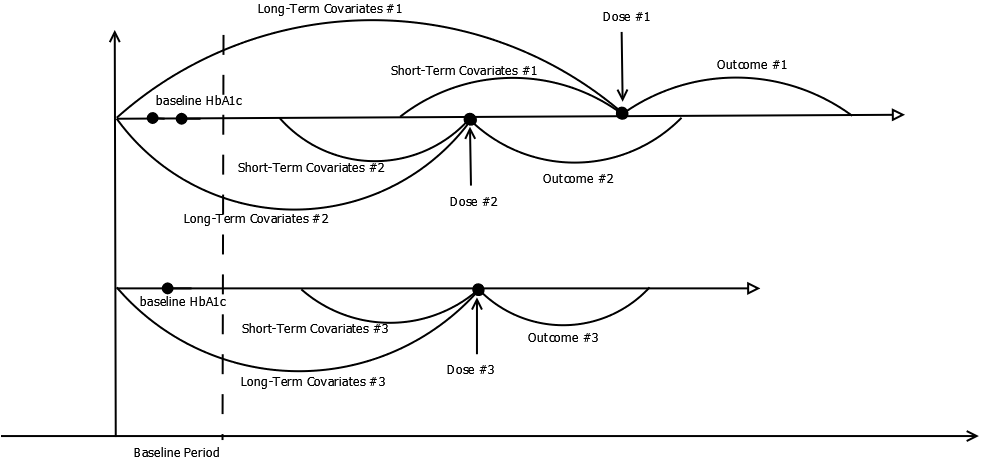}
	\caption{Illustration of Data Structure} \label{fig:realdata-structure}
\end{figure}

The covariates in the data set fall in the following categories.

\begin{enumerate}
	\item Demographics: age, gender, race
	
	\item Long-term measurements: height, the indicator of disability at baseline, length of enrollment, average of HbA1c in the past 91-900 days, the indicator of Medicaid, number of measurements recorded in past the 91-900 days,  Sulfonylureas intake, Insulin intake, and other medications intake 
	
	\item Short-term measurements: number of HbA1c measurements in the past 1-90 days; average of HbA1c in the past 1-90 days, number of hospitalization, number of ED-visits; low-density lipoprotein cholesterol, systolic blood pressure, diastolic blood pressure, BMI, weight, cardiovascular disease, ischemic heart disease, congestive heart failure, hepatocellular carcinoma,  hypoglycemia and injury, infections; and 18 other comorbidities.
	
\end{enumerate} 

\end{document}